\documentclass[twocolumn]{IEEEtran}
\usepackage{amsmath,amssymb,amsthm,epsfig,color}
\usepackage{cite, flushend}
\usepackage{empheq}
\usepackage{stfloats}
\newtheorem{theorem}{Theorem}
\newtheorem{lemma}{Lemma}
\newtheorem{proposition}{Proposition}

\newtheorem*{remark}{Remark}

\begin{document}

%\title{Secrecy SWIPT in MISOME Systems with Probabilistic Constraints}
\title{\LARGE Probabilistically Robust SWIPT for Secrecy MISOME Systems}

%\author{Muhammad R. A. Khandaker,~\IEEEmembership{Member,~IEEE}, and Kai-Kit Wong,~\IEEEmembership{Senior~Member,~IEEE}

\author{Muhammad R. A. Khandaker, Kai-Kit Wong, Yangyang Zhang,\\and Zhongbin Zheng\thanks{M. R. A. Khandaker and K.-K. Wong are with the Department of Electronic and Electrical Engineering, University College London, WC1E 7JE, United Kingdom (e-mail: $\rm m.khandaker@ucl.ac.uk$; $\rm kai\text{-}kit.wong@ucl.ac.uk$).}\thanks{Y. Zhang is with Kuang-Chi Institute of Advanced Technology, Shenzhen, China (e-mail: $\rm yangyang.zhang@kuang\text{-}chi.org$).}\thanks{Z. Zheng is with East China Institute of Telecommunications, China Academy of Information and Communications Technology, Shanghai, China (e-mail: $\rm ben@ecit.org.cn$).}\thanks{This work is supported by EPSRC under grant EP/K015893/1.}}

\maketitle
\begin{abstract}
This paper considers simultaneous wireless information and power transfer (SWIPT) in a multiple-input single-output (MISO) downlink system consisting of one multi-antenna transmitter, one single-antenna information receiver (IR), multiple multi-antenna eavesdroppers (Eves) and multiple single-antenna energy-harvesting receivers (ERs). The main objective is to keep the probability of the legitimate user's achievable secrecy rate outage as well as the ERs' harvested energy outage caused by channel state information (CSI) uncertainties below some prescribed thresholds. As is well known, the secrecy rate outage constraints present a significant analytical and computational challenge. Incorporating the energy harvesting (EH) outage constraints only intensifies that challenge. In this paper, we address this challenging issue using convex restriction approaches which are then proved to yield rank-one optimal beamforming solutions. Numerical results reveal the effectiveness of the proposed schemes.
\end{abstract}

\section{Introduction}
\IEEEPARstart{R}{esearchers} have long been investigating conventional ambient energy resources, such as solar and wind, for energising low-power electronic devices. However, the sporadic and unpredictable nature of these ambient sources makes energy harvesting critical for applications where quality-of-service (QoS) is of priority, and most of these conventional harvesting technologies are only applicable in certain environments and/or weather conditions \cite{krikidis_swipt_mag}. On the other hand, wireless power transfer via magnetic induction is not yet a viable means for widespread applications due to its extremely short distance of power transfer.

Recently, simultaneous wireless information and power transfer (SWIPT) has received enormous interest triggered by the prospects of powering energy-limited wireless devices via RF energy harvesting (EH) due to reductions in power requirements of electronics \cite{swipt_bc, swipt_archi, swipt_oppor, swipt_robust}. As RF signals transport information and energy sumultaneously, mobile users are actually blessed with access to both energy and data at the same time through SWIPT. Surprisingly, this fact did not attract much attention until as late as the last decade.

A practical challenge for SWIPT is that information and energy receiver (ER) circuits operate at very different power sensitivity level (e.g., $-10${\rm dBm} for ERs versus $-60${\rm dBm} for information receivers (IRs)). To store a useful amount of energy from harvesting, the ERs need access to a higher receive power. Due to this incompatibility between the two forms of receivers, two practical schemes, namely, time switching (TS) and power splitting (PS) have been proposed in the literature in order to enable SWIPT \cite{swipt_bc, swipt_archi, jrnl_swipt}. In particular, the authors of \cite{swipt_bc} investigated both schemes in a multiple-input multiple-output (MIMO) broadcasting scenario from a base station (BS) to two mobile receivers decoding information and harvesting energy at the same time. Although the TS scheme simplifies the receiver design, it compromises the efficiencies of the SWIPT technology since the receivers decode information and harvest energy in alternating time slots. Thus, in \cite{swipt_archi}, only the PS-based receiver architecture has been rigorously studied for SWIPT in a point-to-point system. The PS-based schemes have also been considered for multiple-input single-output (MISO) SWIPT multicasting systems with perfect as well as imperfect channel state information (CSI) in \cite{jrnl_swipt} and for MIMO multicasting systems in \cite{conf_swipt}.
%The authors proposed joint multicast transmit beamforming and receive PS algorithms for minimizing the transmit power of the BS subject to quality-of-service (QoS) constraints at each receiver.

Nevertheless, due to channel fading, wireless power transfer efficiency decays drastically as the transmitter-receiver distance increases. Since multiple-antenna techniques provide additional degrees of freedom (DoF) exploiting spatial diversity, installing multiple antennas can help combat channel fading in order to improve wireless power transfer efficiency \cite{swipt_bc}.
%in SWIPT systems

In order to further improve EH efficiency, a receiver-location based {\em near-far} scheduling scheme has also been proposed in the SWIPT literature for information and energy transmissions \cite{swipt_bc, swipt_mu}, in which the receivers only in closer vicinity to the transmitter are scheduled for harvesting energy. Although the scheme apparently seems to be beneficial for EH, actually it gives rise to an undesired security vulnerability for transmitting secret information in scenarios where ERs are supposed to be kept in the dark about the secret message. Specifically, ERs in the scheme have better fading channels than IRs and thus have higher probability to successfully decode the information sent to the IRs \cite{rui_secrecy_swipt, jrnl_secrecy}. On the other hand, the transmitter often needs to apply higher power in order to satisfy ERs' EH requirements which they normally do not need for information-only transmission. This makes secret messages susceptible to external eavesdropping attack as well. Information secrecy can be further degraded if the ERs start cooperating in order to perform joint decoding in an attempt to improve their interception. Thus, the SWIPT systems need to be carefully designed in order to be able to successfully transfer secret messages keeping the potential eavesdroppers ignorant of the secret message to the IR.

%guarantee information secrecy such that the legitimate user (i.e., the IR) can correctly decode the confidential information, but the eavesdroppers (i.e., the ERs) retrieve almost nothing from their observations.

To make sure that the message is delivered secretly to the IR in a SWIPT system taking possible eavesdropping by the ERs into consideration, MISO secrecy beamforming schemes were proposed in \cite{rui_secrecy_swipt, jrnl_secrecy, jrnl_secrecy_sinr}. It was assumed in  \cite{rui_secrecy_swipt} that the ERs do not collude to perform joint decoding and that the CSI of all the nodes, including the eavesdroppers, was perfectly known at the transmitter. Nonetheless, in order to guarantee maximum information secrecy, it would be more meaningful to consider the worst-case scenario, where the ERs collude together to attempt to decode the data jointly such that the eavesdropping rate is maximized. Also, obtaining the eavesdroppers'  CSI perfectly is practically impossible. Hence the authors in \cite{jrnl_secrecy, jrnl_secrecy_sinr} considered robust design based on deterministic channel uncertainty models for SWIPT in scenarios where the ERs may collude together to perform joint decoding in an attempt to improve their interception. In \cite{sec_swipt_ofdma, sec_swipt_ofdma2}, physical layer security was studied for SWIPT in OFDMA networks.

Other works with secrecy in SWIPT either considered worst-case robust approaches in which the CSI errors are assumed to be within a bounded set, or correlation-based approaches in which the channel statistics is available. A delay-limited secrecy SWIPT system was considered in \cite{rui_fading_secrecy} with single-antenna nodes, while a randomization-guided rank-one suboptimal solution was proposed in \cite{wc_swipt_lb} for worst-case MISO secrecy SWIPT systems, and in \cite{wc_swipt_mimo} for MIMO SWIPT systems. Worst-case based MISO secrecy SWIPT optimization has also been considered in \cite{zheng_chu_swipt} for norm-bounded channel uncertainty model based on successive convex approximation approach.

Unfortunately, due to inaccurate channel estimation methods, it may not always be possible that the legitimate transmitter obtains these deterministic models perfectly \cite{arman_tit_sinr}. In such cases, secrecy as well as EH outage may occur. Hence, this paper considers robust secrecy optimization problems with probabilistic secrecy rate and EH constraints for MISO systems with multiple multi-antenna eavesdroppers (MISOME). Two probabilistic secrecy beamforming design problems have been considered namely (1) minimizing the transmit power subject to probabilistic QoS constraints, and (2) maximizing the secrecy rate subject to the total transmit power and secrecy and EH outage constraints. In contrast to the deterministic (worst- or average-case) models \cite{rui_secrecy_swipt, jrnl_secrecy, jrnl_secrecy_sinr, rui_fading_secrecy, wc_swipt_lb, wc_swipt_mimo, zheng_chu_swipt}, this approach offers a safe performance, guaranteeing a certain chance of successful QoS deliveries.

Our main objective is to maintain the probability of the legitimate user's achievable secrecy rate outage as well as the ERs' harvested energy outage caused by CSI uncertainties below given thresholds. As is already well known, the secrecy rate outage constraints present a significant analytical and computational challenge. Inclusion of the EH outage constraints only proliferates that challenge. Unfortunately, quadratic chance constraints generally do not have closed-form expressions, and are unlikely to be tractable. Therefore, it is quite common in the robust optimization literature to develop safe tractable approximations of the outage-based QoS constraints that are computationally efficient and are good in accuracies. Motivated by this, in this paper, we present three conservative approximation approaches namely, Bernstein-type inequality (BTI), $\mathcal{S}$-procedure, and large deviation inequality (LDI) \cite{bti, ben_tal_sproc, wk_ma_outage, zheng_chu_miso_out} based approaches in order to transform the probabilistic constraints into safe and tractable ones.

Applying semidefinite relaxation (SDR) technique, we show that the safe approximation results always yield rank-one optimal transmit covariance solution for the IR. Our simulation results also reveal that the BTI-based restriction approach is the best and the conventional $\mathcal{S}$-procedure based approach is comparatively the worst strategy in terms of secrecy rates, whereas the LDI-based approach appears to be a good compromise. Note that \cite{secrecy_swipt_rob} considered secure robust transmit beamforming design for minimizing the total transmit power under QoS constraints including SINR and SINR outage constraints at the eavesdroppers. The nonconvex probabilistic constraint was replaced with a tractable convex deterministic constraint by bounding the radius of the uncertainty region (comparable to our ${\mathcal S}$-procedure based approach). However, \cite{secrecy_swipt_rob} did not explicitly consider the secrecy rate or EH outage constraints, which are much more challenging to deal with but of great importance to SWIPT systems.
%, i.e., transmit beamforming is optimal for the IR

The rest of this paper is organized as follows. In Section~\ref{sec_sys}, the system model of a secret MISO SWIPT network is introduced. The secrecy rate constrained (SRC) power minimization problem is discussed in Section~\ref{sec_rpm} for the imperfect CSI case whereas in Section~\ref{sec_srm}, solutions to the secrecy rate maximization (SRM) problem are derived. Section~\ref{sec_sim} presents the simulation results that justify the significance of the proposed algorithms under various scenarios. Concluding remarks are provided in Section~\ref{sec_con}.

\emph{Notations}---Throughout this paper, we use the following notations. Boldface lowercase and uppercase letters are used to represent vectors and matrices, respectively. The symbol ${\bf I}_n$ denotes an $n\times n$ identity matrix, $\bf 0$ is a zero vector or matrix. Also, ${\bf A}^H$, ${\rm tr}({\bf A})$, ${\rm rank}({\bf A})$, and $|{\bf A}|$ represent the Hermitian (conjugate) transpose, trace, rank and determinant of a matrix ${\bf A}$; ${\rm Pr}[\cdot]$ represents the probability of an event;  $\|\cdot\|$ and $\|\cdot\|_{\rm F}$ represent the Euclidean norm and Frobenius norm, respectively; ${\bf A}\succeq {\bf 0}\, ({\bf A}\succ {\bf 0})$ means that ${\bf A}$ is a Hermitian positive semidefinite (definite) matrix. The notation ${\bf x}\sim \mathcal{CN}(\boldsymbol{\mu}, {\boldsymbol\Sigma})$ means that ${\bf x}$ is a random vector following a complex circularly symmetric Gaussian distribution with mean vector $\boldsymbol{\mu}$ and covariance matrix ${\boldsymbol\Sigma}$.

\begin{figure}[ht]
\centering
\includegraphics*[width=7cm]{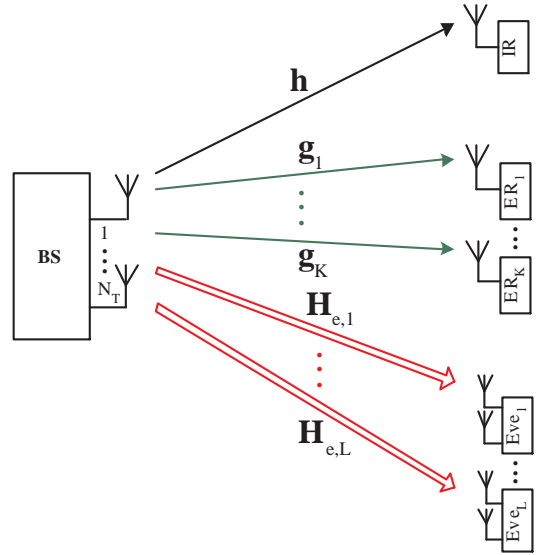}
\caption{A MISO SWIPT system with multiple-antenna eavesdroppers.}\label{sysmod}
\end{figure}

\section{System Model}\label{sec_sys}
A MISO downlink system is considered for SWIPT with $K + L + 1$ receivers as shown in Fig.~\ref{sysmod}. The BS transmits information to a legitimate user and energy to $K$ EH receivers keeping the information as secure as possible from $L$ eavesdroppers (Eves). It is assumed that the BS performs necessary user selection operation before transmission begins for decoding information. Note that similar assumption has also been made in numerous existing works in the secrecy literature for SWIPT, e.g., \cite{rui_secrecy_swipt, jrnl_secrecy, jrnl_secrecy_sinr, rui_fading_secrecy, wc_swipt_lb, wc_swipt_mimo, zheng_chu_swipt}. The transmitter or BS has $N_{\rm T} > 1$ transmitting antennas and each legitimate receiver (IR or ER) has single receiving antenna. Eves are equipped with $N_{{\rm e},i}$ antennas for $i = 1, \dots, L$. The BS performs linear transmit beamforming to send secret information to the IR. We assume that the ERs are also legitimate users of the network authorized for EH only and do not intend to overhear the message destined to the IR.\footnote{The results of this paper are readily extendible to study the impacts of untrusted ERs on the SWIPT system performance.} By letting ${\bf x}$ be the transmit signal vector, the received signals at the IR, the $k$th ER, and the $i$th Eve can be modeled, respectively, as
\begin{align}
y_{\rm d}&={\bf h}^H{\bf x} + n_{\rm d},\\
y_{{\rm h},k}&={\bf g}_{k}^H{\bf x} + n_{{\rm h},k}, ~\mbox{for }k=1,\dots, K,\\
{\bf y}_{{\rm e},i}&={\bf H}_{{\rm e},i}^H{\bf x} + {\bf n}_{{\rm e},i}, ~\mbox{for }i=1,\dots, L,
\end{align}
where ${\bf h}$, ${\bf g}_k$, and ${\bf H}_{{\rm e},i}$ are the conjugated complex channel vector (matrix) between the BS and the IR, the $k$th ER, and the $i$th Eve, respectively, $n_{\rm d}\sim\mathcal{CN}(0,\sigma_{\rm d}^2)$, $n_{{\rm h},k}\sim\mathcal{CN}(0,\sigma_{\rm h}^2)$, and ${\bf n}_{{\rm e},l}\sim\mathcal{CN}({\bf 0},\sigma_{\rm e}^2{\bf I}_{N_l})$ are the additive Gaussian noises at the IR, the $k$th ER, and the $i$th Eve, respectively. For notational simplicity, the path loss factors $P_{f,d} = \sqrt{L_cd_{\rm I}^{-\bar\kappa}}$, $P_{f,hk} = \sqrt{L_cd_{{\rm h},k}^{-\bar\kappa}}$, and $P_{f,ei} = \sqrt{L_cd_{{\rm e},i}^{-\bar\kappa}}$, with the path loss constant $L_c = G_TG_R\left(\frac{c}{4\pi f}\right)^2$ (in which $G_T$ and $G_R$ are the antenna gains of the transmitter and the receivers, respectively, $c$ is the speed of light, and $f$ is the carrier frequency), $d_{\rm I}$, $d_{{\rm h},k}$, and $d_{{\rm e},i}$ indicating the distance of the IR, $k$th ER, and $i$th Eve, respectively, from the BS, and $\bar\kappa$ (typically between $2.7$ and $3.5$) being the path loss exponent, are all assumed to be absorbed in the corresponding channel gains. The BS chooses ${\bf x}$ as ${\bf x}={\bf b}_{\rm I}s_{\rm I}$ where $s_{\rm I}\sim\mathcal{CN}(0,1)$ is the confidential information-bearing signal for the IR and ${\bf b}_{\rm I}$ is the transmit beamforming vector. Hence, ${\bf x}$ is a beamformed version of the message.

Note that multi-antenna eavesdroppers may interchangeably mean that single-antenna eavesdroppers located in a favourable location to cooperate may collude together to improve their interception. For ease of exposition, we further assume that all the Eves are colluding into multiple groups. In particular, Eves are assumed to perform joint maximum signal-to-noise ratio (SNR) receive beamforming. By denoting ${\bf Q}_{\rm I}\triangleq {\bf b}_{\rm I}{\bf b}_{\rm I}^H$ as the transmit covariance matrix, the mutual information (MI) between the BS and the IR is given by
\begin{equation}\label{C_I}
C_{\rm I}\left({\bf Q}_{\rm I}\right) = \log \left(1 + \frac{1}{\sigma_{\rm d}^2}{\bf h}^H{\bf Q}_{\rm I}{\bf h}\right),
\end{equation}
and that between the BS and the colluded Eves is given by
\begin{equation}\label{C_E}
C_{{\rm e},i}\left({\bf Q}_{\rm I}\right) = \log \left|{\bf I}_{N_{{\rm e},i}} + \frac{1}{\sigma_{\rm e}^2}{\bf H}_{{\rm e},i}^H{\bf Q}_{\rm I}{\bf H}_{{\rm e},i}\right|, ~\mbox{for }i = 1, \dots, L.
\end{equation}
Given ${\bf Q}_{\rm I}$, the achievable secrecy rate is given by \cite{goel_an}
\begin{equation}\label{cs}
C_{\rm s} = \min_i \left\{C_{\rm I}\left({\bf Q}_{\rm I}\right) - C_{{\rm e},i}\left({\bf Q}_{\rm I}\right)\right\}^+,
\end{equation}
where $\{a\}^+=\max(0,a)$. %Note that \eqref{cs} gives the perfect secrecy rate when the IR can correctly decode the confidential information at $C_{\rm s}$ bits per channel use, while the Eves can retrieve almost nothing about the secret message.
The harvested power at the $k$th ER is given by
\begin{equation}
E_k = \xi_k{\bf g}_{k}^H{\bf Q}_{\rm I}{\bf g}_{k},\label{eh}
\end{equation}
where $\xi_k\in(0,1]$ is the energy conversion efficiency of the energy transducers at the $k$th ER. %that accounts for the loss in the energy transducers for converting the harvested RF energy to electrical energy. It is worth pointing out that the ERs do not need to convert the received signal from the RF band to the baseband in order to harvest the carried energy using modern energy transducers. Therefore, according to the law of energy conservation, it is assumed that the total harvested RF band power (energy normalized by the baseband symbol period) at each ER is proportional to the normalised energy of the received baseband signal.
For simplicity, it is assumed that the harvested energy due to the background noise at the EH receivers in \eqref{eh} is negligible and as a consequence can be ignored \cite{swipt_bc}.

In most of the existing works with secrecy for SWIPT, it is assumed that the instantaneous CSI of all the receivers is available at the transmitter. However, in practical wireless communication systems, perfect CSI is likely not available and an important issue is how to robustify a secure transmit design in the presence of imperfect CSI. As a consequence, our next exertion is to develop convex optimization algorithms that satisfy given chance constraints exploiting only imperfect CSI knowledge.

\section{Robust Power Minimization with Probabilistic Constraints}\label{sec_rpm}
In this section, we develop probabilistically robust algorithm for the SRC problem with secrecy rate and EH outage constraints. The interest here is in active Eves cases, where the Eves themselves are also users of the network and the transmitter aims to provide different services to different types of users. That is, the Eves are legitimate users for utilities other than the particular information destined to the IR. For such active eavesdroppers, the CSI can be estimated from the eavesdroppers' transmission. Thus we assume that the BS has incomplete knowledge of the ERs' as well as Eves' channels while the IR's channel is perfectly known. The perfect IR's CSI assumption is quite widely exercised in the existing literature since the legitimate IR's CSI may be obtained at very high precision through secure control channels \cite{rui_secrecy_swipt, jrnl_secrecy, wc_swipt_mimo}.

We consider the commonly used Gaussian channel error model for the imperfect CSI. To model the imperfect CSI, we assume that the actual channels ${\bf g}_{k}$, for $k = 1, \dots, K,$ lie in the neighbourhood of the estimated channels $\hat{\bf g}_{k}$, for $k = 1, \dots, K,$ available at the transmitter. The channel error vectors are assumed to have circularly symmetric complex Gaussian (CSCG) distribution. Thus, the actual channels are modeled as
\begin{equation}
{\bf g}_{k} = \hat{\bf g}_{k} + {\boldsymbol\delta}_{{\rm g},k}, ~\mbox{for }k = 1,\dots, K,
\end{equation}
where $\hat{\bf g}_{k} \in\mathbb{C}^{N_{\rm T}\times 1}$ is the estimated CSI of the $k$th ER and ${\boldsymbol\delta}_{{\rm g},k} \in\mathbb{C}^{N_{\rm T}\times 1}$, for $k = 1, \dots, K,$ represent the channel uncertainties such that ${\boldsymbol\delta}_{{\rm g},k} \sim {\mathcal {CN}}\left({\bf 0},{\bf R}_{{\rm g},k}\right)$, in which ${\bf R}_{{\rm g},k}$ is a positive semidefinite matrix. Similarly, the Eves' channel uncertainty model can be represented by
\begin{equation}
{\bf H}_{{\rm e},i} = \hat{\bf H}_{{\rm e},i} + {\boldsymbol\Delta}_{{\rm H},i}, ~\mbox{for }i = 1,\dots, L,
\end{equation}
where $\hat{\bf H}_{{\rm e},i} \in\mathbb{C}^{N_{\rm T}\times N_{{\rm e},i}}$ is the estimated CSI of the $i$th Eve and ${\boldsymbol\Delta}_{{\rm H},i} \in\mathbb{C}^{N_{\rm T}\times N_{{\rm e},i}}$, for $i = 1, \dots, L,$ represent the channel uncertainties such that ${\boldsymbol\delta}_{{\rm H},i} \triangleq {\rm vec}\left({\boldsymbol\Delta}_{{\rm H},i}\right) \sim {\mathcal {CN}}\left({\bf 0},{\bf R}_{{\rm H},i}\right)$, where ${\bf R}_{{\rm H},i}$ is a positive semidefinite matrix. Thus, the probabilistically robust SRC power minimization problem can be formulated as
\begin{subequations}\label{PrConst1}
\begin{eqnarray}
\min_{{\bf Q}_{\rm I}} \!\!\!& &\!\!\! {\rm tr} \left({\bf Q}_{\rm I}\right)\label{PrConst1_o}\\
{\rm s.t.} \!\!\!& &\!\!\!\!\! {\rm Pr}\!\left[\min_i \left\{\!C_{\rm I}\!\left({\bf Q}_{\rm I}\right) \!-\! \hat{C}_{{\rm e},i}\!\left({\bf Q}_{\rm I}\right)\!\right\}^+\!\!\!\ge\!\! R\right]\!\!\ge\! 1 \!-\! p,\!\forall i,\label{PrConst1_c1}\\
\!\!\!& &\!\!\! {\rm Pr}\left[\min_k \hat E_k \geq \eta_k\right] \ge 1 - q, \forall k,\label{PrConst1_c2}\\
\!\!\!& &\!\!\! {\bf Q}_{\rm I} \succeq {\bf 0}, \label{PrConst1_c3}
\end{eqnarray}
\end{subequations}
where $\hat C_{{\rm e},i}\left({\bf Q}_{\rm I}\right) = \log \left|{\bf I}_{N_{{\rm e},i}} + \frac{1}{\sigma_{\rm e}^2} {\bf H}_{{\rm e},i}^H{\bf Q}_{\rm I} {\bf H}_{{\rm e},i}\right|$ is the $i$th Eves' MI with ${\bf H}_{{\rm e},i} = \hat{\bf H}_{{\rm e},i} + {\boldsymbol\Delta}_{{\rm H},i}, ~\mbox{for }i = 1,\dots, L, {\boldsymbol\delta}_{{\rm H},i} \triangleq {\rm vec}\left({\boldsymbol\Delta}_{{\rm H},i}\right) \sim {\mathcal {CN}}\left(0,{\bf R}_{{\rm H},i}\right)$ and $\hat E_k = \xi_k{\bf g}_{k}^H{\bf Q}_{\rm I}{\bf g}_{k}$ is the average energy harvested with ${\bf g}_{k} = \hat{\bf g}_{k} + {\boldsymbol\delta}_{{\rm g},k}$, for $k = 1,\dots, K, {\boldsymbol\delta}_{{\rm g},k} \sim {\mathcal {CN}}\left({\bf 0},{\bf R}_{{\rm g},k}\right)$. The problem formulation in \eqref{PrConst1} guarantees that the IR can successfully decode its message at least $(1 - p)\times 100\%$ of the time. Similarly, the ERs can harvest the minimum required amount of power at least $(1 - q)\times 100\%$ of the time.
% \textcolor{red}{(in order to prove the rank-one solution, we must assume $\sigma_{\rm h}^2 = 0$, which is commonly assumed!)}.
As the desired outage probability decreases, the size of the feasible sets described by \eqref{PrConst1_c1} and \eqref{PrConst1_c2} decreases. Hence, one might expect an increase in the required transmit power with decreasing outage probabilities.

Notice that the rank constraint on ${\bf Q}_{\rm I}$ has been relaxed in problem \eqref{PrConst1}. %\textcolor{red}{The benefit of rank relaxation is that the inequalities inside the probability function in now linear in ${\bf Q}_{\rm I}$ which makes the probabilistic constraints more tractable.}
An important issue that arises from the relaxation is the rank of the resulting solution. The removal of the rank constraint ${\rm rank}({\bf Q}_{\rm I}) = 1$ means that the solution obtained through solving problem \eqref{PrConst1} may have rank higher than one. A common practice of overcoming this is to apply some rank approximation procedure (e.g., randomization) to the optimal ${\bf Q}_{\rm I}^*$ to find a feasible beamforming solution ${\bf b}_{\rm I}$. However, in this paper, we aim to prove the tightness of the rank relaxation.

The problem is still nonconvex due to the probabilistic constraints involving $\log\det$ functions. To make those constraints more tractable, we introduce the following lemma.

\begin{lemma}[\cite{qli_sdp}]\label{lem_tr}
For any positive semidefinite matrix ${\bf A}$, the following inequality holds
\begin{equation}\label{lem_det}
|{\bf I}+{\bf A}|\geq 1+{\rm tr}({\bf A})
\end{equation}
and the equality in \eqref{lem_det} holds if and only if ${\rm rank}({\bf A}) \leq 1$.
\end{lemma}

\begin{proof}
Let $r_A = {\rm rank}({\bf A})$. While the case of $r_A = 0$ is trivial, for $r_A \ge 1$, let $\lambda_1 \ge \lambda_2 \ge \cdots \lambda_{r_A} > 0$ denote the nonzero eigenvalues of ${\bf A}$. Accordingly, we have that
\begin{align}
|{\bf I}+{\bf A}| = \prod_{i=1}^{r_A}(1+\lambda_{i}) & = 1+\sum_{i=1}^{r_A}\lambda_{i}+\sum_{i\neq k}\lambda_{i}\lambda_{k} +\ldots \nonumber\\
& \geq 1+\sum_{i=1}^{r_A}\lambda_{i}=1+{\rm Tr}({\bf A}). \nonumber
\end{align}
Clearly, the above equality holds if and only if $r_A = 1$. %\hfill $\blacksquare$
\end{proof}

Now, by applying \emph{Lemma~1}, the secrecy rate outage constraint \eqref{PrConst1_c1} can be relaxed as
%\begin{subequations}
\begin{multline}
{\rm Pr}\left[\log \left(1 + \frac{1}{\sigma_{\rm d}^2}{\bf h}^H{\bf Q}_{\rm I}{\bf h}\right) - \log \left(1 + \frac{1}{\sigma_{\rm e}^2}{\rm tr}\left( {\bf H}_{{\rm e},i}^H{\bf Q}_{\rm I}\right.\right.\right.\\
\left.\left.\left.\times {\bf H}_{{\rm e},i}\right)\right) \ge R\right] \ge 1 - p,\forall i, \nonumber
\end{multline}
which is equivalent to
\begin{multline}
{\rm Pr}\left[\left(1 + \frac{1}{\sigma_{\rm d}^2}{\bf h}^H{\bf Q}_{\rm I}{\bf h}\right)\ge {2^R}\left(1 + \frac{1}{\sigma_{\rm e}^2}{\rm tr}\left( {\bf H}_{{\rm e},i}^H{\bf Q}_{\rm I}\right.\right.\right.\\
\left.\left.\left.\times {\bf H}_{{\rm e},i}\right)\right)\right]
\ge 1 - p,\forall i. \nonumber
\end{multline}
Rearranging the terms in the above equation yields
\begin{multline}
{\rm Pr}\left[{\rm tr}\left( {\bf H}_{{\rm e},i}^H{\bf Q}_{\rm I} {\bf H}_{{\rm e},i}\right) \le \frac{\sigma_{\rm e}^2}{2^R}\left(1 + \frac{1}{\sigma_{\rm d}^2}{\bf h}^H{\bf Q}_{\rm I}{\bf h}\right) - {\sigma_{\rm e}^2}\right] \\
\ge 1 - p,\forall i. \label{sec_out_prob}
\end{multline}
%\end{subequations}
Replacing ${\bf H}_{{\rm e},i} = \hat{\bf H}_{{\rm e},i} + {\boldsymbol\Delta}_{{\rm H},i}$, and then performing some mathematical manipulations, we finally obtain from \eqref{sec_out_prob}

\begin{multline}
{\rm Pr}\left[{\rm tr}\left({\boldsymbol\Delta}_{{\rm H},i}^H{\bf Q}_{\rm I}{\boldsymbol\Delta}_{{\rm H},i} + {\boldsymbol\Delta}_{{\rm H},i}^H{\bf Q}_{\rm I}\hat {\bf H}_{{\rm e},i} + \hat {\bf H}_{{\rm e},i}^H{\bf Q}_{\rm I}{\boldsymbol\Delta}_{{\rm H},i}\right.\right.\\
\left.\left. + \hat {\bf H}_{{\rm e},i}^H{\bf Q}_{\rm I}\hat {\bf H}_{{\rm e},i}\right) \le \frac{\sigma_{\rm e}^2}{2^R}\left(1 + \frac{1}{\sigma_{\rm d}^2}{\bf h}^H{\bf Q}_{\rm I}{\bf h}\right) - {\sigma_{\rm e}^2}\right]\\
\ge 1 - p,\forall i. \label{sec_out_prob2}
\end{multline}

Recall that the relaxation \eqref{sec_out_prob} is in fact tight according to \emph{Lemma~\ref{lem_tr}} if ${\rm rank}\big({\bf Q}_{\rm I}\big) \leq 1$. Our goal is to reformulate problem \eqref{PrConst1} as a tractable convex problem and then prove that the relaxation in \eqref{sec_out_prob} is indeed tight for the chance-constrained secrecy problem by proving the rank-one structure of ${\bf Q}_{\rm I}$. To make the robust problem \eqref{PrConst1} more tractable to analyze and solve, we first transform the robust constraints in (\ref{PrConst1_c1}) and (\ref{PrConst1_c2}) into convex inequalities using advanced matrix inequality results in the optimization literature.

Note that the probability term in \eqref{sec_out_prob2} does not have a closed-form expression. Now we apply the following matrix identities to reformulate the secrecy outage constraint
\begin{subequations}\label{mat_id}
\begin{align}
{\rm vec}\left({\bf AXB}\right) & = \left({\bf B}^H \otimes{\bf A}\right){\rm vec}\left({\bf X}\right),\label{mat_id1}\\
{\rm tr}\left({\bf A}^H{\bf B}\right) & = {\rm vec}\left({\bf A}\right)^H{\rm vec}\left({\bf B}\right),\label{mat_id2}\\
\left({\bf A}\otimes{\bf B}\right)^H & = {\bf A}^H\otimes{\bf B}^H.\label{mat_id3}
\end{align}
\end{subequations}
Applying the identities in \eqref{mat_id}, we can express \eqref{sec_out_prob2} as
\begin{multline}\label{sec_out_prob3}
{\rm Pr}\left[{\boldsymbol\delta}_{{\rm H},i}^H\left({\bf I}_{N_{{\rm e},i}}\!\otimes\!{\bf Q}_{\rm I}\right){\boldsymbol\delta}_{{\rm H},i} + 2\Re\left\{{\boldsymbol\delta}_{{\rm H},i}^H\left({\bf I}_{N_{{\rm e},i}}\!\otimes\!{\bf Q}_{\rm I}\right)\hat {\bf h}_{{\rm e},i}\right\}\right.\\
\left. + \hat {\bf h}_{{\rm e},i}^H\left({\bf I}_{N_{{\rm e},i}}\otimes{\bf Q}_{\rm I}\right)\hat {\bf h}_{{\rm e},i}\le \frac{\sigma_{\rm e}^2}{2^R}\left(1 + \frac{1}{\sigma_{\rm d}^2}{\bf h}^H{\bf Q}_{\rm I}{\bf h}\right) - {\sigma_{\rm e}^2}\right]\\
\ge 1 \! - \! p,\forall i,
\end{multline}
where $\hat {\bf h}_{{\rm e},i} \triangleq {\rm vec}(\hat {\bf H}_{{\rm e},i})$. Since ${\boldsymbol\delta}_{{\rm H},i} \sim {\mathcal {CN}}\left({\bf 0},{\bf R}_{{\rm H},i}\right)$, ${\boldsymbol\delta}_{{\rm H},i}$ can be reexpressed as ${\boldsymbol\delta}_{{\rm H},i} = {\bf R}_{{\rm H},i}^{\frac{1}{2}}{\bf v}_{{\rm H},i}$ such that ${\bf v}_{{\rm H},i} \sim {\mathcal {CN}}\left({\bf 0},{\bf I}_{N_{\rm T}N_{{\rm e},i}}\right)$. Thus, \eqref{sec_out_prob3} can be reexpressed as
\begin{multline}
{\rm Pr}\left[{\bf v}_{{\rm H},i}^H\left[-{\bf R}_{{\rm H},i}^{\frac{1}{2}}\left({\bf I}_{N_{{\rm e},i}}\!\otimes\!{\bf Q}_{\rm I}\right){\bf R}_{{\rm H},i}^{\frac{1}{2}}\right]{\bf v}_{{\rm H},i} + 2\Re\left\{{\bf v}_{{\rm H},i}^H\right.\right.\\
\left.\times\left[-{\bf R}_{{\rm H},i}^{\frac{1}{2}}\left({\bf I}_{N_{{\rm e},i}}\!\otimes\!{\bf Q}_{\rm I}\right)\hat {\bf h}_{{\rm e},i}\right]\right\} + \frac{\sigma_{\rm e}^2}{2^R}\left(1 + \frac{1}{\sigma_{\rm d}^2}{\bf h}^H{\bf Q}_{\rm I}{\bf h}\right)\\
\left.- {\sigma_{\rm e}^2} - \hat {\bf h}_{{\rm e},i}^H\left({\bf I}_{N_{{\rm e},i}}\otimes{\bf Q}_{\rm I}\right)\hat {\bf h}_{{\rm e},i} \ge 0\right] \ge 1- p,\forall i. \label{sec_out_prob4}
\end{multline}
Similarly, the EH outage constraint can be expressed as
\begin{multline}\label{eh_out1}
{\rm Pr}\left[{\bf u}_{{\rm g},k}^H{\bf R}_{{\rm g},k}^{\frac{1}{2}}{\bf Q}_{\rm I}{\bf R}_{{\rm g},k}^{\frac{1}{2}}{\bf u}_{{\rm g},k} + 2\Re\left\{{\bf u}_{{\rm g},k}^H{\bf R}_{{\rm g},k}^{\frac{1}{2}}{\bf Q}_{\rm I}\hat {\bf g}_{k}\right\}\right.\\
\left.+ \hat {\bf g}_{k}^H{\bf Q}_{\rm I}\hat {\bf g}_{k} - \frac{\eta_k}{\xi_k} \ge 0\right] \ge 1 - q,\forall i,
\end{multline}
where ${\boldsymbol\delta}_{{\rm g},k} = {\bf R}_{{\rm g},k}^{\frac{1}{2}}{\bf u}_{{\rm g},k}$ with ${\bf u}_{{\rm g},k} \sim {\mathcal {CN}}\left({\bf 0},{\bf I}_{N_{\rm T}}\right)$. In the following, we tackle these probabilistic constraints pursuing some convex restriction approaches. Clearly, the constraints \eqref{sec_out_prob4} and \eqref{eh_out1} are of the form:
$${\rm Pr}\{{\bf e}^{H} {\bf Q} {\bf e} + 2 {\Re}\{{\bf e}^{H} {\bf r}\} + s \geq 0 \} \geq 1 - \rho.$$
If we can find a convex function $f({\bf Q}, {\bf r},s)$, such that
$${\rm Pr}\{{\bf e}^{H} {\bf Q} {\bf e} + 2 {\Re}\{{\bf e}^{H} {\bf r}\} + s \geq 0 \} \le f(\bf {Q}, \bf {r},s),$$
then we will readily have the implication \cite{wk_ma_outage}
\begin{multline}
f({\bf Q}, {\bf r},s) \leq \rho \\
\Longrightarrow {\rm Pr}\{{\bf e}^{H} {\bf Q} {\bf e} + 2 {\Re}\{{\bf e}^{H} {\bf r}\} + s \geq 0 \} \geq 1 - \rho. \label{impli_bti0}
\end{multline}
Hence, the L.H.S. of the implication in \eqref{impli_bti0} gives a safe approximation, which is convex, of the generally intractable probabilistic constraint in the R.H.S. In the following, we attempt to derive the convex restrictions $f(\bf {Q}, \bf {r},s)$ for tackling the probabilistic constraints. The derived restriction methods differ in terms of both computational complexity and tightness.

\subsection{Robust Optimization Based on BTI}
The relaxation step alone does not provide a convex approximation of the original problem. The semidefinite probabilistic constraints remain intractable. In order to make the secrecy and EH outage constraints more tractable, we consider BTI in this subsection. The Bernstein-type concentration inequalities play a central role to transform the probabilistic constraints into more tractable form based on large deviation inequality for complex Gaussian quadratic vector functions. For completeness, the inequality theorem is presented in {\em Lemma \ref{lemm_bti}} below.

\begin{lemma}[\bfseries BTI] \label{lemm_bti}
Consider the chance constraint
\begin{eqnarray}
{\rm Pr}\left[{\bf x}^{H}{\bf A}{\bf x}+2\Re\left\{{\bf x}^{H}{\bf r}\right\} + \theta \ge 0\right]\ge 1 - \rho,\label{chance_const}
\end{eqnarray}
where ${\bf x}$ is a standard complex Gaussian random vector, i.e., ${\bf x} \sim {\mathcal {CN}}\left({\bf 0},{\bf I}_n\right)$, the $3$-tuple $({\bf A}, {\bf r}, \theta)$ (${{\bf A} \in {\mathbb H}^{n\times n}}$ is a complex hermitian matrix, ${\bf r} \in {\mathbb C}^{n},$ ${b} \in {\mathbb R}$) forms a set of deterministic optimization variables, and $\rho \in (0,\, 1]$ is fixed. The following implication holds \cite{bti}:

\begin{subequations}\label{impli_bti}
${\rm Pr}\left[{\bf x}^{H}{\bf A}{\bf x}+2\Re\left\{{\bf x}^{H}{\bf r}\right\} + \theta \ge 0\right] \ge 1 - \rho$
\begin{empheq}[left={\Longleftarrow}\empheqlbrace]{align}
%\begin{align}
{\rm tr}\left({\bf A}\right) - \sqrt{-2\ln(\rho)}\psi + \ln(\rho)\omega + \theta & \ge 0,\\
\left\|\left[\begin{array}{c}{\rm vec}\left({\bf A}\right)\\ \sqrt{2}{\bf r}\end{array}\right]\right\| & \le \psi,\\
\omega{\bf I}_n + {\bf A} \succeq {\bf 0},\quad
\psi, \omega & \ge 0,
%\end{align}
\end{empheq}
\end{subequations}
where $\psi, \omega \in \mathbb{R}$ are slack variables.
\end{lemma}

Note that \eqref{impli_bti} is jointly convex in $\bf A$, $\psi$, and $\omega$ and represents an efficiently computable convex restriction of  the chance constraint \eqref{chance_const}. Thus using BTI, we can derive closed-form upper bounds on the violation probability to construct an efficiently computable convex function.

Indeed, for each constraint in \eqref{sec_out_prob4}, the following correspondence can be shown for all $i$:
\begin{subequations}
\begin{align}
{\bf A}_i & = -{\bf R}_{{\rm H},i}^{\frac{1}{2}}\left({\bf I}_{N_{{\rm e},i}}\otimes{\bf Q}_{\rm I}\right){\bf R}_{{\rm H},i}^{\frac{1}{2}},\\
{\bf r}_i & = -{\bf R}_{{\rm H},i}^{\frac{1}{2}}\left({\bf I}_{N_{{\rm e},i}}\otimes{\bf Q}_{\rm I}\right)\hat {\bf h}_{{\rm e},i},\\
\theta_i & \!=\! \frac{\sigma_{\rm e}^2}{2^R}\!\left(\!1 \!+\! \frac{1}{\sigma_{\rm d}^2}{\bf h}^H{\bf Q}_{\rm I}{\bf h}\!\right) \!-\! {\sigma_{\rm e}^2} \!-\! \hat {\bf h}_{{\rm e},i}^H\!\left({\bf I}_{N_{{\rm e},i}}\!\otimes\!{\bf Q}_{\rm I}\right)\!\hat {\bf h}_{{\rm e},i}.
\end{align}
\end{subequations}
Based on {\em Lemma \ref{lemm_bti}}, a convex safe approximation of the secrecy rate constraint in \eqref{sec_out_prob4} can be equivalently represented by
\begin{subequations}\label{impli_bti_sec}
%\begin{empheq}[left={\Leftarrow}\empheqlbrace]{align}
\begin{align}
&{\rm tr}\left({\bf R}_{{\rm H},i}^{\frac{1}{2}}\left({\bf I}_{N_{{\rm e},i}}\otimes{\bf Q}_{\rm I}\right){\bf R}_{{\rm H},i}^{\frac{1}{2}}\right) + \hat {\bf h}_{{\rm e},i}^H\left({\bf I}_{N_{{\rm e},i}}\otimes{\bf Q}_{\rm I}\right)\hat {\bf h}_{{\rm e},i}- \frac{\sigma_{\rm e}^2}{2^R} \nonumber\\
& \!\!+\! \sqrt{-2\ln(p)}\psi_i \!-\! \ln(p)\omega_i
 - \frac{\sigma_{\rm e}^2}{2^R\sigma_{\rm d}^2}{\bf h}^H{\bf Q}_{\rm I}{\bf h} + {\sigma_{\rm e}^2} \le 0,\\
& \left\|\left[\begin{array}{c}{\rm vec}\left({\bf R}_{{\rm H},i}^{\frac{1}{2}}\left({\bf I}_{N_{{\rm e},i}}\otimes{\bf Q}_{\rm I}\right){\bf R}_{{\rm H},i}^{\frac{1}{2}}\right)\\ \sqrt{2}{\bf R}_{{\rm H},i}^{\frac{1}{2}}\left({\bf I}_{N_{{\rm e},i}}\otimes{\bf Q}_{\rm I}\right)\hat {\bf h}_{{\rm e},i}\end{array}\right]\right\| \le \psi_i,\\
& \omega_i{\bf I}_{N_{\rm T}N_{{\rm e},i}}\!\! - {\bf R}_{{\rm H},i}^{\frac{1}{2}}\left({\bf I}_{N_{{\rm e},i}}\!\otimes{\bf Q}_{\rm I}\right){\bf R}_{{\rm H},i}^{\frac{1}{2}} \succeq {\bf 0}, \psi_i, \omega_i \!\ge\! 0, \!\!\forall i,
\end{align}
%\end{empheq}
\end{subequations}
and that of the EH outage constraint \eqref{eh_out1} can be equivalently recast into
\begin{subequations}\label{impli_bti_eh}
%\begin{empheq}[left={\Leftarrow}\empheqlbrace]{align}
\begin{align}
&{\rm tr}\left({\bf R}_{{\rm g},k}^{\frac{1}{2}}{\bf Q}_{\rm I}{\bf R}_{{\rm g},k}^{\frac{1}{2}}\right) + \hat {\bf g}_{k}^H{\bf Q}_{\rm I}\hat {\bf g}_{k} + \sqrt{-2\ln(q)}\nu_k - \ln(q)\varphi_k\nonumber\\
&\qquad\qquad\qquad\qquad\qquad\qquad\qquad\qquad - \frac{\eta_k}{\xi_k} \ge 0,\\ % + {\sigma_{\rm h}^2}
& \left\|\left[\begin{array}{c}{\rm vec}\left({\bf R}_{{\rm g},k}^{\frac{1}{2}}{\bf Q}_{\rm I}{\bf R}_{{\rm g},k}^{\frac{1}{2}}\right)\\ \sqrt{2}{\bf R}_{{\rm g},k}^{\frac{1}{2}}{\bf Q}_{\rm I}\hat {\bf g}_{k}\end{array}\right]\right\| \le \nu_k,\\
& \varphi_k{\bf I}_{N_{\rm T}} + {\bf R}_{{\rm g},k}^{\frac{1}{2}}{\bf Q}_{\rm I}{\bf R}_{{\rm g},k}^{\frac{1}{2}} \succeq {\bf 0}, ~~ \nu_k, \varphi_k \ge 0, \forall k,
\end{align}
%\end{empheq}
\end{subequations}
where $q \in (0,\, 1]$ is fixed, and $\nu, \varphi \in \mathbb{R}$ are slack variables. Incorporating \eqref{impli_bti_sec} and \eqref{impli_bti_eh}, the power minimization problem \eqref{PrConst1} can be equivalently reformulated as
\begin{subequations}\label{PrConst2}
\begin{align}
& \min_{{\bf Q}_{\rm I}, \{\psi_i\}, \{\omega_i\}, \{\nu_k\}, \{\varphi_k\}} ~~ ~~ {\rm tr} \left({\bf Q}_{\rm I}\right) \quad {\rm s.t.} \label{PrConst2_o}\\
~~ & ~~ {\rm tr}\left({\bf R}_{{\rm H},i}^{\frac{1}{2}}\left({\bf I}_{N_{{\rm e},i}}\!\otimes\!{\bf Q}_{\rm I}\right){\bf R}_{{\rm H},i}^{\frac{1}{2}}\right) + \hat {\bf h}_{{\rm e},i}^H\left({\bf I}_{N_{{\rm e},i}}\otimes{\bf Q}_{\rm I}\right)\hat {\bf h}_{{\rm e},i}\nonumber\\
& \!+\! \sqrt{\!-2\ln(p)}\psi_i \!-\! \ln(p)\omega_i \!-\! \frac{\sigma_{\rm e}^2}{2^R\sigma_{\rm d}^2}{\bf h}^H{\bf Q}_{\rm I}{\bf h} \!-\! \frac{\sigma_{\rm e}^2}{2^R} \!+\! {\sigma_{\rm e}^2} \!\le\! 0,\label{PrConst2_c1}\\
& ~~ \left\|\left[\begin{array}{c}{\rm vec}\left({\bf R}_{{\rm H},i}^{\frac{1}{2}}\left({\bf I}_{N_{{\rm e},i}}\!\otimes\!{\bf Q}_{\rm I}\right){\bf R}_{{\rm H},i}^{\frac{1}{2}}\right)\\ \sqrt{2}{\bf R}_{{\rm H},i}^{\frac{1}{2}}\left({\bf I}_{N_{{\rm e},i}}\!\otimes\!{\bf Q}_{\rm I}\right)\hat {\bf h}_{{\rm e},i}\end{array}\right]\right\|  \le \psi_i,\label{PrConst2_c2}\\
& ~~ \omega_i{\bf I}_{N_{\rm T}N_{{\rm e},i}} - {\bf R}_{{\rm H},i}^{\frac{1}{2}}\left({\bf I}_{N_{{\rm e},i}}\!\otimes\!{\bf Q}_{\rm I}\right){\bf R}_{{\rm H},i}^{\frac{1}{2}} \succeq {\bf 0},\label{PrConst2_c3}\\
& ~~ {\rm tr}\left({\bf R}_{{\rm g},k}^{\frac{1}{2}}{\bf Q}_{\rm I}{\bf R}_{{\rm g},k}^{\frac{1}{2}}\right) + \hat {\bf g}_{k}^H{\bf Q}_{\rm I}\hat {\bf g}_{k} - \sqrt{-2\ln(q)}\nu_k\nonumber\\
&\qquad\qquad\qquad\qquad\qquad + \ln(q)\varphi_k - \frac{\eta_k}{\xi_k} \ge 0,\label{PrConst2_c4}\\% + {\sigma_{\rm h}^2}
& ~~ \left\|\left[\begin{array}{c}{\rm vec}\left({\bf R}_{{\rm g},k}^{\frac{1}{2}}{\bf Q}_{\rm I}{\bf R}_{{\rm g},k}^{\frac{1}{2}}\right)\\ \sqrt{2}{\bf R}_{{\rm g},k}^{\frac{1}{2}}{\bf Q}_{\rm I}\hat {\bf g}_{k}\end{array}\right]\right\| \le \nu_k,\label{PrConst2_c5}\\
& ~~ \varphi_k{\bf I}_{N_{\rm T}} + {\bf R}_{{\rm g},k}^{\frac{1}{2}}{\bf Q}_{\rm I}{\bf R}_{{\rm g},k}^{\frac{1}{2}} \succeq {\bf 0} ,\label{PrConst2_c6}\\
& ~~ {\bf Q}_{\rm I} \succeq {\bf 0}, \quad \psi_i, \omega_i \ge 0, \forall i, \quad \nu_k, \varphi_k \ge 0, \forall k.\label{PrConst2_c7}
%\eqref{impli_bti_sec}, ~~ \eqref{impli_bti_eh}, ~~ {\bf Q}_{\rm I} \succeq {\bf 0}. \label{PrConst2_c2}
\end{align}
\end{subequations}
The problem \eqref{PrConst2} is convex and can be efficiently solved using interior-point based solvers \cite{cvx}. %Interestingly, the following theorem states that the relaxed problem \eqref{PrConst2} always yields a rank-one transmit beamforming solution.

\begin{theorem}\label{thm_rank_bti}
Suppose that the relaxed problem \eqref{PrConst2} is feasible for $R>0$. The optimal solution must satisfy ${\rm rank}\big({\bf Q}_{\rm I}\big) = 1$.
\end{theorem}

\begin{proof}
See Appendix~\ref{proof_thm_rank_bti}.
\end{proof}

%%%%%%%%%%%%Single column equation%%%%%%%%%%%%%%%%%%%%%%%%%%%%%%%%%%%%%%
\begin{figure*}[t]
\setcounter{equation}{30}
\begin{equation} \label{sec_sproc_lmi}
{\boldsymbol\Theta}_i\left({\bf Q}_{\rm I},\mu_{{\rm H},i}\right) \triangleq \left[\begin{array}{cc}\mu_{{\rm H},i}{\bf I}_{N_{\rm T}N_{{\rm e},i}} -{\bf R}_{{\rm H},i}^{\frac{1}{2}}\left({\bf I}_{N_{{\rm e},i}}\!\otimes\!{\bf Q}_{\rm I}\right){\bf R}_{{\rm H},i}^{\frac{1}{2}} & -{\bf R}_{{\rm H},i}^{\frac{1}{2}}\left({\bf I}_{N_{{\rm e},i}}\!\otimes\!{\bf Q}_{\rm I}\right)\hat {\bf h}_{{\rm e},i}\\
-\hat {\bf h}_{{\rm e},i}^H\left({\bf I}_{N_{{\rm e},i}}\!\otimes\!{\bf Q}_{\rm I}\right){\bf R}_{{\rm H},i}^{\frac{1}{2}} & \tau_{{\rm H},i} - \mu_{{\rm H},i}\gamma_{{\rm e}}^2\end{array}\right]\succeq{\bf 0},
\end{equation}

\begin{equation} \label{eh_sproc_lmi}
{\boldsymbol\Upsilon}_k\left({\bf Q}_{\rm I},\mu_{{\rm g},k}\right) \triangleq \left[\begin{array}{cc}\mu_{{\rm g},k}{\bf I}_{N_{\rm T}} + {\bf R}_{{\rm g},k}^{\frac{1}{2}}{\bf Q}_{\rm I}{\bf R}_{{\rm g},k}^{\frac{1}{2}} & {\bf R}_{{\rm g},k}^{\frac{1}{2}}{\bf Q}_{\rm I}\hat {\bf g}_{k}\\
\hat {\bf g}_{k}^H{\bf Q}_{\rm I}{\bf R}_{{\rm g},k}^{\frac{1}{2}} & \hat {\bf g}_{k}^H{\bf Q}_{\rm I}\hat {\bf g}_{k} - \frac{\eta_k}{\xi_k} - \mu_{{\rm g},k}\gamma_{{\rm e}}^2\end{array}\right]\succeq{\bf 0},
\end{equation}
\hrulefill \normalsize
\end{figure*}
\setcounter{equation}{24}
%%%%%%%%%%%%%%%%%%%%%%%%%%%%%%%%%%%%%%%%%%%%%%%%%%%%%%%%%%%%%%%%%%%%%%%%

\subsection{${\mathcal S}$-Procedure Based Approach}\label{sub_sec_sproc}
In this subsection, we develop a convex restriction approach in a conservative fashion for robust optimization. The main idea is to choose a set for the channel uncertainty region satisfying the probabilistic restriction. In contrast to norm-bounded CSI errors, we have the freedom to choose the set arbitrarily in this method according to the maximum tolerable outage probability. Towards this end, the following lemma is useful.
\begin{lemma}\label{lemm_sph_bound}
Consider an arbitrary set $\mathcal{A} \subset \mathbb{C}^{N_{\rm T}\times 1}$ satisfying ${\rm Pr}\left\{{\bf x} \in \cal{A}\right\} \ge 1 - \rho$. The following implication holds \cite{wk_ma_outage}:
\begin{multline}
{\bf x}^{H} {\bf A} {\bf x} + 2 {\Re}\{{\bf x}^{H} {\bf r} \} + \theta \geq 0, \forall {\bf x} \in {\cal A},\\
\Longrightarrow {\rm Pr}\{{\bf x}^{H} {\bf Q} {\bf x} + 2 {\Re}\{{\bf x}^{H} {\bf r}\} + \theta \geq 0 \} \geq 1 - \rho. \label{sph_bound}
\end{multline}
\end{lemma}
That is, the worst-case robust constraint on the L.H.S.~of \eqref{sph_bound} is a safe approximation of the probabilistic constraint on the R.H.S. Based on {\em Lemma \ref{lemm_sph_bound}}, given the following deterministic quadratic constraint (from \eqref{sec_out_prob4})
\begin{multline}
{\bf v}_{{\rm H},i}^H\left[-{\bf R}_{{\rm H},i}^{\frac{1}{2}}\left({\bf I}_{N_{{\rm e},i}}\otimes{\bf Q}_{\rm I}\right){\bf R}_{{\rm H},i}^{\frac{1}{2}}\right]{\bf v}_{{\rm H},i}
+ 2\Re\left\{{\bf v}_{{\rm H},i}^H\left[-{\bf R}_{{\rm H},i}^{\frac{1}{2}}\right.\right.\\
\left.\left.\times\left({\bf I}_{N_{{\rm e},i}}\otimes{\bf Q}_{\rm I}\right)\right]\hat {\bf h}_{{\rm e},i}\right\} + \frac{\sigma_{\rm e}^2}{2^R}\left(1 + \frac{1}{\sigma_{\rm d}^2}{\bf h}^H{\bf Q}_{\rm I}{\bf h}\right) - {\sigma_{\rm e}^2}\\
- \hat {\bf h}_{{\rm e},i}^H\left({\bf I}_{N_{{\rm e},i}}\otimes{\bf Q}_{\rm I}\right)\hat {\bf h}_{{\rm e},i} \ge 0,\forall i, \label{sec_out_sproc1}
\end{multline}
choosing the following set for the channel uncertainty region
\begin{eqnarray}
\mathcal{A} = \{{\bf v}_{{\rm H},i} \in \mathbb{C}^{N_{\rm T}\times 1}| {\rm Pr}\left({\bf v}_{{\rm H},i}^H{\bf v}_{{\rm H},i} \le \gamma_{{\rm e}}^2\right) \ge 1- p\}, \forall i, \label{uncert_sec}
\end{eqnarray}
is sufficient to guarantee the probabilistic constraint in \eqref{sec_out_prob4}. Interestingly, with ${\bf v}_{{\rm H},i}$ defined as ${\bf v}_{{\rm H},i} \sim {\mathcal {CN}} \left(0, {\bf I}_{N_{\rm T}N_{{\rm e},i}}\right)$, it can be easily verified that $\|{\bf v}_{{\rm H},i}\|^2$ is a Chi-square ($\chi^2$) random variable with $2N_{\rm T}N_{{\rm e},i}$ degrees of freedom. The channel uncertainty region in \eqref{uncert_sec} always holds for $\gamma_{{\rm e}} = \sqrt{\frac{\mathcal{F}_{\chi_{m}^2}^{-1}(1-p)}{2}}$, where $\mathcal{F}_{\chi_{m}^2}^{-1}(a)$ is the inverse cumulative distribution function (CDF) of the Chi-square random variable $a$ with $m = 2N_{\rm T}N_{{\rm e},i}$ degrees of freedom. In fact, $\gamma_{{\rm e},i}$ can be interpreted as the radius of the ball $\mathcal{A}$ defining the channel uncertainty region. Thus the probabilistic constraint \eqref{sec_out_prob4} can be equivalently expressed by the following set of inequalities:
\begin{subequations}\label{sec_out_sproc2}
%\begin{empheq}[left=\empheqlbrace]{align}
\begin{align}
&{\bf v}_{{\rm H},i}^H\left[-{\bf R}_{{\rm H},i}^{\frac{1}{2}}\left({\bf I}_{N_{{\rm e},i}}\!\otimes\!{\bf Q}_{\rm I}\right){\bf R}_{{\rm H},i}^{\frac{1}{2}}\right]{\bf v}_{{\rm H},i}
+ 2\Re\left\{{\bf v}_{{\rm H},i}^H\left[-{\bf R}_{{\rm H},i}^{\frac{1}{2}}\right.\right.\nonumber\\
&\left.\left.\times\left({\bf I}_{N_{{\rm e},i}}\!\otimes\!{\bf Q}_{\rm I}\right)\right]\hat {\bf h}_{{\rm e},i}\right\} + \frac{\sigma_{\rm e}^2}{2^R}\left(1 + \frac{1}{\sigma_{\rm d}^2}{\bf h}^H{\bf Q}_{\rm I}{\bf h}\right)- {\sigma_{\rm e}^2}\nonumber\\
&\qquad\qquad - \hat {\bf h}_{{\rm e},i}^H\left({\bf I}_{N_{{\rm e},i}}\otimes{\bf Q}_{\rm I}\right)\hat {\bf h}_{{\rm e},i} \ge 0,\forall i,\\
&\qquad\qquad\quad\qquad-{\bf v}_{{\rm H},i}^H{\bf v}_{{\rm H},i} + \gamma_{{\rm e}}^2 \ge 0,\forall i.
\end{align}
%\end{empheq}
\end{subequations}

At this point, we apply the so-called ${\mathcal S}$-procedure \cite{boyd} to transform the constraint \eqref{sec_out_sproc2} into a more tractable linear matrix inequality (LMI). The ${\mathcal S}$-procedure is presented in {\em Lemma \ref{lemm_sproc}} below.

\begin{lemma}[\bfseries$\boldsymbol{\mathcal S}$-Procedure] \label{lemm_sproc}
Let $f_{i}({\bf x}), i = 1, 2,$ be defined as
\begin{equation}
f_{i}({\bf x})={\bf x}^{H}{\bf A}_{i}{\bf x}+2\Re\left\{{\bf b}_{i}^{H}{\bf x}\right\}+c_{i},
\end{equation}
where ${{\bf A}_{i} \in {\mathbb C}^{n\times n}}, {\bf b}_{i} \in {\mathbb C}^{n}, {c}_{i} \in {\mathbb R}$. The implication $f_{1}({\bf x}) \leq 0 \Rightarrow f_{2}({\bf x}) \leq 0$ holds if and only if there exists $\mu \geq 0$ such that
\begin{equation}
\mu\left[\begin{array}{cc}
{\bf A}_{1}&{\bf b}_{1}\\{\bf b}_{1}^{H}&c_{1}
\end{array}\right]-\left[\begin{array}{cc}
{\bf A}_{2}&{\bf b}_{2}\cr{\bf b}_{2}^{H}&c_{2}
\end{array}\right]\succeq{\bf 0}
\end{equation}
provided that there exists a point $\hat{{\bf x}}$ such that $f_1(\hat{{\bf x}})<0$.
\end{lemma}

\setcounter{equation}{32}

According to {\em Lemma \ref{lemm_sproc}}, \eqref{sec_out_sproc2} holds if and only if there exists $\mu_{{\rm H},i} \geq 0, \forall i,$ such that ${\boldsymbol\Theta}_i\left({\bf Q}_{\rm I},\mu_{{\rm H},i}\right) \succeq{\bf 0}$ (defined in \eqref{sec_sproc_lmi} at the top of the page)
%\begin{equation} \label{sec_sproc_lmi}
%{\boldsymbol\Theta}_i\left({\bf Q}_{\rm I},\mu_{{\rm H},i}\right) \succeq{\bf 0},
%\end{equation}
where $\tau_{{\rm H},i} \triangleq \frac{\sigma_{\rm e}^2}{2^R}\left(1 + \frac{1}{\sigma_{\rm d}^2}{\bf h}^H{\bf Q}_{\rm I}{\bf h}\right) - {\sigma_{\rm e}^2} - \hat {\bf h}_{{\rm e},i}^H\left({\bf I}_{N_{{\rm e},i}}\otimes{\bf Q}_{\rm I}\right)\hat {\bf h}_{{\rm e},i}, \forall i$.

Similarly, the EH outage constraint \eqref{eh_out1} can be transformed to the LMI in \eqref{eh_sproc_lmi}
%following LMI:
%\begin{equation} \label{eh_sproc_lmi}
%{\boldsymbol\Upsilon}_k\left({\bf Q}_{\rm I},\mu_{{\rm g},k}\right) \succeq{\bf 0},
%\end{equation}
where $\mu_{{\rm g},k} \geq 0, \forall k$. By exploiting ${\mathcal S}$-Procedure, the power minimization problem \eqref{PrConst1} can be equivalently reformulated as
\begin{subequations}\label{PrConstS}
\begin{align}
\min_{{\bf Q}_{\rm I}, \{\mu_{{\rm H},i}\}, \{\mu_{{\rm g},k}\}} ~~ & ~ {\rm tr} \left({\bf Q}_{\rm I}\right)\label{PrConstS_o}\\
{\rm s.t.} ~~ & ~ {\boldsymbol\Theta}_i\left({\bf Q}_{\rm I},\mu_{{\rm H},i}\right)\succeq{\bf 0}, \forall i, \label{PrConstS_c1}\\
%\left[\begin{array}{cc}\mu_{{\rm H},i}{\bf I}_{N_{\rm T}N_{{\rm e},i}} -{\bf R}_{{\rm H},i}^{\frac{1}{2}}\left({\bf I}_{N_{{\rm e},i}}\!\otimes\!{\bf Q}_{\rm I}\right){\bf R}_{{\rm H},i}^{\frac{1}{2}} & -{\bf R}_{{\rm H},i}^{\frac{1}{2}}\left({\bf I}_{N_{{\rm e},i}}\!\otimes\!{\bf Q}_{\rm I}\right)\hat {\bf h}_{{\rm e},i}\\
%-\hat {\bf h}_{{\rm e},i}^H\left({\bf I}_{N_{{\rm e},i}}\!\otimes\!{\bf Q}_{\rm I}\right){\bf R}_{{\rm H},i}^{\frac{1}{2}} & \tau_{{\rm H},i} - \mu_{{\rm H},i}\gamma_{{\rm e}}^2\end{array}\right]\succeq{\bf 0}, \forall i, \label{PrConstS_c1}\\
& ~ {\boldsymbol\Upsilon}_k\left({\bf Q}_{\rm I},\mu_{{\rm g},k}\right)\succeq{\bf 0}, \forall k,\label{PrConstS_c2}\\
%\left[\begin{array}{cc}\mu_{{\rm g},k}{\bf I}_{N_{\rm T}} + {\bf R}_{{\rm g},k}^{\frac{1}{2}}{\bf Q}_{\rm I}{\bf R}_{{\rm g},k}^{\frac{1}{2}} & {\bf R}_{{\rm g},k}^{\frac{1}{2}}{\bf Q}_{\rm I}\hat {\bf g}_{k}\\
%\hat {\bf g}_{k}^H{\bf Q}_{\rm I}{\bf R}_{{\rm g},k}^{\frac{1}{2}} & \hat {\bf g}_{k}^H{\bf Q}_{\rm I}\hat {\bf g}_{k} - \frac{\eta_k}{\xi_k} - \mu_{{\rm g},k}\gamma_{{\rm e}}^2\end{array}\right]\succeq{\bf 0}, \forall k,\label{PrConstS_c2}\\
& ~{\bf Q}_{\rm I} \!\succeq\! {\bf 0}, ~ \mu_{{\rm H},i} \!\ge\! 0, \forall i, ~ \mu_{{\rm g},k} \ge 0, \forall k. \label{PrConstS_c3}
\end{align}
\end{subequations}
The SDP problem \eqref{PrConstS} is convex and can be efficiently solved using interior-point based solvers \cite{cvx}. Interestingly, the following theorem states that the relaxed problem \eqref{PrConstS} always yields a rank-one transmit beamforming solution.

\begin{theorem}\label{thm_rank_sproc}
Suppose that the relaxed problem \eqref{PrConstS} is feasible for $R>0$. The optimal solution must satisfy ${\rm rank}\big({\bf Q}_{\rm I}\big) = 1$.
\end{theorem}

\begin{proof}
See Appendix~\ref{proof_thm_rank_sproc}.
\end{proof}

\begin{remark}
At this point, we would like to comment that the BTI-based approach has higher computational complexity compared to the $\mathcal{S}$-procedure based approach as the former involves a more compound mixture of different types of constraints \cite{wk_ma_outage}. However, we will perform a full complexity analysis later in this section.
\end{remark}

\subsection{LDI Based Approach}
Since the BTI and $\mathcal{S}$-procedure based approaches transform the chance-constrained optimization problem \eqref{PrConst1} into SDPs, the resulting safe designs are polynomial-time solvable \cite{boyd}. The SDPs can, however, be very expensive to solve if the size of the LMI constraints in \eqref{PrConst2} and \eqref{PrConstS} is sufficiently large. Hence in this subsection, our endeavour is to develop convex restrictions involving simpler conic constraints. The method follows from the decomposition-based LDI \cite{ldi} for complex Gaussian quadratic functions as defined in the following lemma.

\begin{lemma}\label{lemm_ldi}
\cite[Lemma~2]{wk_ma_outage}
Let ${\bf x} \sim {\mathcal {CN}}\left({\bf 0},{\bf I}_n\right)$ be a standard complex Gaussian random vector, and let ${{\bf A} \in {\mathbb H}^{n\times n}}$ and ${\bf r} \in {\mathbb C}^{n}$ be given. Then, for any $v > \frac{1}{\sqrt{2}}$ and $\zeta > 0$, we have
\begin{align}\,\,\,{\rm Pr}& \left \{{\bf x}^{H} {\bf A} {\bf x} + 2 {\rm Re}\{{\bf x}^{H} {\bf r}\} \leq {\rm tr} ({\bf A}) - \zeta \right \}\nonumber\\
& \leq \begin{cases}\exp \left (- {\frac{\zeta ^{2}} {4T^{2}}} \right)& \!\!for ~ 0 < \zeta \leq 2\bar {v}vT,\\
 \exp \left (- {\frac{\bar {v }v\zeta } {T}} + (\bar {v }v)^{2} \right) & \!\!for ~ \zeta > 2\bar {v }vT,\end{cases}
\end{align}
where $\bar {v} = 1- {\frac{1} {2v^{2}}}$ and $T = v\Vert {\bf A}\Vert _{F} + {\frac{1}{\sqrt {2}}}\Vert {\bf r}\Vert$.
\end{lemma}

The merit of {\em Lemma~\ref{lemm_ldi}} is that it helps decomposing a sum of dependent random variables into sums of independent random variables. This idea has been used extensively in the literature of probability theory; see, e.g., \cite{ldi, lmi_sdp}.

Next, we concentrate on deriving convex restrictions of \eqref{sec_out_prob4} and \eqref{eh_out1} based on the LDI approach using {\em Lemma~\ref{lemm_ldi}}. For equation \eqref{sec_out_prob4}, we set
\begin{align}
\zeta_i & = {\rm tr}\left(-{\bf R}_{{\rm H},i}^{\frac{1}{2}}\left({\bf I}_{N_{{\rm e},i}}\otimes{\bf Q}_{\rm I}\right){\bf R}_{{\rm H},i}^{\frac{1}{2}}\right) + \tau_{{\rm H},i},\\
T_{{\rm H},i} & = v\Vert -{\bf R}_{{\rm H},i}^{\frac{1}{2}}\left({\bf I}_{N_{{\rm e},i}}\otimes{\bf Q}_{\rm I}\right){\bf R}_{{\rm H},i}^{\frac{1}{2}}\Vert _{F}\nonumber\\
&\qquad + {\frac{1}{\sqrt {2}}}\Vert -{\bf R}_{{\rm H},i}^{\frac{1}{2}}\left({\bf I}_{N_{{\rm e},i}}\otimes{\bf Q}_{\rm I}\right)\hat {\bf h}_{{\rm e},i}\Vert.
\end{align}
Here $v$ is obtained from the solution to the following quadratic equation
\begin{equation}
\bar{v}v = (1-1/(2v^2))v = \sqrt{-\ln(p)}\label{quad_eqn}
\end{equation}
such that $v > \frac{1}{\sqrt{2}}$. It has been shown in \cite{wk_ma_outage} that such a $v$ must always exist, since $(1-1/(2v^2))v = 0$ when $v = \frac{1}{\sqrt{2}}$ and $(1-1/(2v^2))v$ is a monotonically increasing function of $v$ within the interval $[\frac{1}{\sqrt{2}},\infty)$.
Now from \eqref{quad_eqn}, we conclude that $p = \exp\left (-(\bar{v}v)^{2} \right)$. Furthermore, according to {\em Lemma~\ref{lemm_ldi}}, the chance constraint \eqref{sec_out_prob4} will be satisfied if we choose $\zeta_i = 2\sqrt{-\ln(p)}T_{{\rm H},i}$ for the interval $2\sqrt{-\ln(p)}T_{{\rm H},i} \le \zeta_i \le 2\bar{v}vT_{{\rm H},i}$. On the other hand, if $\zeta_i > 2\bar {v }vT_{{\rm H},i} = 2\sqrt{-\ln(p)}T_{{\rm H},i}$, then {\em Lemma~\ref{lemm_ldi}} yields
\begin{multline}
{\rm Pr}\left\{{\bf v}_{{\rm H},i}^H\left[-{\bf R}_{{\rm H},i}^{\frac{1}{2}}\left({\bf I}_{N_{{\rm e},i}}\otimes{\bf Q}_{\rm I}\right){\bf R}_{{\rm H},i}^{\frac{1}{2}}\right]{\bf v}_{{\rm H},i}\right.\\
\left. + 2\Re\left\{{\bf v}_{{\rm H},i}^H\left[-{\bf R}_{{\rm H},i}^{\frac{1}{2}}\left({\bf I}_{N_{{\rm e},i}}\otimes{\bf Q}_{\rm I}\right)\right]\hat {\bf h}_{{\rm e},i}\right\} + \tau_{{\rm H},i} \leq 0 \right\}\\
\leq \exp \left(- {\frac{\bar{v}v\zeta_i }{T_{{\rm H},i}}} + (\bar{v}v)^{2} \right) < \exp\left (-(\bar{v}v)^{2} \right) = p,
\end{multline}
which essentially indicates that the chance constraint \eqref{sec_out_prob4} will still be satisfied. %Thus we have
%$${\rm Pr}\left \{{\bf x}^{H} {\bf A} {\bf x} + 2 {\Re}\{{\bf x}^{H} {\bf r}\} + \theta < 0 \right \} \leq \exp \left (- {{({\rm tr}({\bf A})+\theta)^{2}}\over {4T^{2}}} \right),$$
%which suggests that we can assume Formula .
The resulting convex restriction can thus be expressed as
\begin{equation}\label{soc1}
\!{\rm tr}(-{\bf R}_{{\rm H},i}^{\frac{1}{2}}\!\left({\bf I}_{N_{{\rm e},i}}\otimes{\bf Q}_{\rm I}\right)\!{\bf R}_{{\rm H},i}^{\frac{1}{2}}) \!+  \tau_{{\rm H},i} \!\geq\! 2\sqrt {-\ln (p)} T_{{\rm H},i}, \!\forall i.
\end{equation}
%Following the definition of $T_{{\rm H},i}$, it can be easily shown that \eqref{soc1} can be expressed as a system of second-order cone (SOC) constraints. In particular, we obtain the following convex restriction method for tackling the probabilistic constraint \eqref{sec_out_prob4}:
Thus using the definition of $T_{{\rm H},i}$, we obtain the following system of second-order cone (SOC) constraints from \eqref{soc1} in order to tackle the probabilistic constraint \eqref{sec_out_prob4}:
%\begin{subequations}
\begin{align}
\begin{cases}
{\rm tr}\left(-{\bf R}_{{\rm H},i}^{\frac{1}{2}}\left({\bf I}_{N_{{\rm e},i}}\otimes{\bf Q}_{\rm I}\right){\bf R}_{{\rm H},i}^{\frac{1}{2}}\right) + \tau_{{\rm H},i} \!\!\!\!\!& \ge\!\! 2\sqrt {-\ln (p)}\\
& \times\left(\bar\psi_i + \bar\omega_i\right),\\
\frac{1}{\sqrt{2}}\left\|-{\bf R}_{{\rm H},i}^{\frac{1}{2}}\left({\bf I}_{N_{{\rm e},i}}\otimes{\bf Q}_{\rm I}\right)\hat {\bf h}_{{\rm e},i}\right\| & \le \bar\psi_i,\\
v\left\|-{\bf R}_{{\rm H},i}^{\frac{1}{2}}\left({\bf I}_{N_{{\rm e},i}}\otimes{\bf Q}_{\rm I}\right){\bf R}_{{\rm H},i}^{\frac{1}{2}}\right\|_F & \le \bar\omega_i,
\end{cases}
\end{align}
%\end{subequations}
where $\bar\psi_i, \bar\omega_i \in \mathbb{R}, \forall i,$ are slack variables. Similarly, defining the slack variables $\bar\nu_k$ and $\bar\varphi_k$, the EH outage constraint \eqref{eh_out1} can be expressed as
%\begin{subequations}
\begin{align}
\begin{cases}
{\rm tr}\left({\bf R}_{{\rm g},k}^{\frac{1}{2}}{\bf Q}_{\rm I}{\bf R}_{{\rm g},k}^{\frac{1}{2}}\right) + \hat {\bf g}_{k}^H{\bf Q}_{\rm I}\hat {\bf g}_{k} - \frac{\eta_k}{\xi_k} & \ge 2\sqrt {-\ln (q)}\\
& \times\left(\bar\nu_k + \bar\varphi_k\right),\\
\frac{1}{\sqrt{2}}\left\|{\bf R}_{{\rm g},k}^{\frac{1}{2}}{\bf Q}_{\rm I}\hat {\bf g}_{k}\right\| & \le \bar\nu_k,\\
v\left\|{\bf R}_{{\rm g},k}^{\frac{1}{2}}{\bf Q}_{\rm I}{\bf R}_{{\rm g},k}^{\frac{1}{2}}\right\|_F & \le \bar\varphi_k.
\end{cases}
\end{align}
%\end{subequations}
By applying the LDI method to the outage constrained problem \eqref{PrConst1}, we obtain the convex restriction formulation of the power minimization problem as
\begin{subequations}\label{PrConstD}
\begin{align}
& \min_{{\bf Q}_{\rm I}, \{\bar\psi_i\}, \{\bar\omega_i\}, \{\bar\nu_k\}, \{\bar\varphi_k\}} ~~  ~~ {\rm tr} \left({\bf Q}_{\rm I}\right) \qquad {\rm s.t.} \label{PrConstD_o}\\
& \qquad {\rm tr}\left(-{\bf R}_{{\rm H},i}^{\frac{1}{2}}\left({\bf I}_{N_{{\rm e},i}}\otimes{\bf Q}_{\rm I}\right){\bf R}_{{\rm H},i}^{\frac{1}{2}}\right) + \tau_{{\rm H},i}\nonumber\\
& \qquad\qquad\qquad\qquad \ge 2\sqrt {-\ln (p)} \!\left(\bar\psi_i \!+ \bar\omega_i\right),\label{PrConstD_c1}\\
& \qquad \frac{1}{\sqrt{2}}\left\|-{\bf R}_{{\rm H},i}^{\frac{1}{2}}\left({\bf I}_{N_{{\rm e},i}}\otimes{\bf Q}_{\rm I}\right)\hat {\bf h}_{{\rm e},i}\right\| \le \bar\psi_i,\label{PrConstD_c2}\\
& \qquad v\left\|{\rm vec}\left(-{\bf R}_{{\rm H},i}^{\frac{1}{2}}\left({\bf I}_{N_{{\rm e},i}}\otimes{\bf Q}_{\rm I}\right){\bf R}_{{\rm H},i}^{\frac{1}{2}}\right)\right\| \le \bar\omega_i, \label{PrConstD_c3}\\
& \qquad {\rm tr}\left({\bf R}_{{\rm g},k}^{\frac{1}{2}}{\bf Q}_{\rm I}{\bf R}_{{\rm g},k}^{\frac{1}{2}}\right) + \hat {\bf g}_{k}^H{\bf Q}_{\rm I}\hat {\bf g}_{k} - \frac{\eta_k}{\xi_k}\nonumber\\
& \qquad\qquad\qquad\qquad  \ge 2\sqrt {-\ln (q)} \left(\bar\nu_k + \bar\varphi_k\right),\label{PrConstD_c4}\\
& \qquad \frac{1}{\sqrt{2}}\left\|{\bf R}_{{\rm g},k}^{\frac{1}{2}}{\bf Q}_{\rm I}\hat {\bf g}_{k}\right\| \le \bar\nu_k, \label{PrConstD_c5}\\
& \qquad v\left\|{\rm vec}\left({\bf R}_{{\rm g},k}^{\frac{1}{2}}{\bf Q}_{\rm I}{\bf R}_{{\rm g},k}^{\frac{1}{2}}\right)\right\| \le \bar\varphi_k,\label{PrConstD_c6}\\
& \qquad {\bf Q}_{\rm I} \succeq {\bf 0}, \quad \bar\psi_i, \bar\omega_i \ge 0, \forall i, \quad \bar\nu_k, \bar\varphi_k \ge 0, \forall k. \label{PrConstD_c7}
\end{align}
\end{subequations}

Since the above convex problem contains only SOC constraints, it can be solved more efficiently than the convex restrictions obtained using BTI based and $\mathcal S$-procedure based approaches. Finally, the following theorem studies the tightness of the rank relaxation in problem \eqref{PrConst1}.

\begin{theorem}\label{thm_rank_ldi}
Suppose that the relaxed problem \eqref{PrConstD} is feasible for $R>0$. The optimal solution must satisfy ${\rm rank}\big({\bf Q}_{\rm I}\big) = 1$.
\end{theorem}

\begin{proof}
See Appendix~\ref{proof_thm_rank_ldi}.
\end{proof}

\begin{remark}
It can be verified that the optimal solutions to problems \eqref{PrConst2}, \eqref{PrConstS}, and \eqref{PrConstD} obtained in Theorems~1, 2, and 3, respectively, are unique. Let us first consider that there are two distinct optimal solutions to problem \eqref{PrConst2}, say ${\bf Q}_1$ and ${\bf Q}_2$ ,such that ${\rm rank}({\bf Q}_1) = {\rm rank}({\bf Q}_2) = 1$. Therefore, the range spaces of ${\bf Q}_1$ and ${\bf Q}_2$ must be different. Following a basic concept in convex optimization, any ${\bf Q}_3 = \lambda{\bf Q}_1 + (1-\lambda){\bf Q}_2$, for $\lambda \in (0,1)$, is also an optimal solution of \eqref{PrConst2} \cite{boyd}. Since ${\bf Q}_1$ and ${\bf Q}_2$ are both rank-one and distinct, it is rigid that ${\bf Q}_3$ is of rank two, which contradicts with the result proved in {\em Theorem~1}. Hence problem \eqref{PrConst2} must have only one optimal solution ${\bf Q}_{\rm I}$. Similarly, the uniqueness of the optimal solutions of problems \eqref{PrConstS}, and \eqref{PrConstD} can also be verified.
\end{remark}

\subsection{Complexity Analysis}
In this subsection, we mathematically characterize the computational complexity of the proposed schemes. Note that the convex restriction formulations \eqref{PrConst2}, \eqref{PrConstS}, and \eqref{PrConstD} involve only LMI and SOC constraints, and hence can be solved using standard interior-point methods (IPM) \cite[Lecture~6]{cvx_opt_lect}. Therefore, we can use the worst-case computation time of IPM to compare the complexities of the three formulations. Now, using the transformation
$${\mathbb H} ^{n} \ni {\bf S} \mapsto \left[\begin{matrix}{\rm Re}({\bf S}) & -{\rm Im}({\bf S}) \cr {\rm Im}({\bf S}) & {\rm Re}({\bf S})\end{matrix}\right] \in {\mathbb S} ^{2n},$$
where ${\mathbb S} ^{n}$ and ${\mathbb H} ^{n}$ represent the sets of $n \times n$ real symmetric matrices and complex Hermitian matrices, respectively, we can convert the complex-valued conic programs \eqref{PrConst2}, \eqref{PrConstS}, and \eqref{PrConstD} into equivalent real-valued conic programs of the form \cite{cplx_sdp, wk_ma_outage}
\begin{align}
\min _{{\bf z}\in {\mathbb R} ^{n}} & \quad {\bf c} ^{T} {\bf z}\cr
{\rm s.t.}& \quad {\sum _{i=1}^{n} z_{i} {\bf A}_{i}^{j} - {\bf B}^{j} \in {\mathbb S}_{+}^{k_{j}} } \qquad {\rm for~} j=1,\ldots, p, \cr & \quad {\bf T} ^{j} {\bf z}- {\bf b}^{j} \in {\mathbb L} ^{k_{j}} \qquad {\rm for~} j=p+1,\ldots,m. \label{ipm}
\end{align}
Here, ${\bf T} ^{j} \in {\mathbb R}^{k_{j}\times n}$, ${\bf b} ^{j} \in {\mathbb R}^{k_{j}}$ for $j=p+1,\dots, m$, ${\bf c}\in {\mathbb R}^{n}$, ${\mathbb S}_{+}^{k}$ is the set of $k \times k$ real PSD matrices, and ${\mathbb L} ^{k}$ is the second-order cone of dimention $k \ge 1$.
Now the overall complexity of the IPM for solving the above problem consists of two components:

\begin{itemize}
    \item[a)] \textit{Iteration Complexity}: The number of iterations required to reach an $\epsilon$-accurate ($\epsilon > 0$) optimal solution of problem \eqref{ipm} is in the order of $\ln(1/\epsilon)\sqrt{\beta(\mathcal{K})}$, where $\beta(\mathcal{K}) = \sum_{j=1}^p k_j + 2(m-p)$ is known to be the barrier parameter.

    \item[b)] \textit{Per-Iteration Computation Cost}: A system of $n$ linear equations is required to be solved in each iteration. The computation tasks include the formation of the coefficient matrix $\bf H$ of the system of linear equations and the factorization of $\bf H$. The cost of forming $\bf H$ sums on the order of $\kappa_{\rm for}=n\sum_{j=1}^p k_j^3 + n^2 \sum_{j=1}^p k_j^2 + n \sum_{j=p+1}^m k_j^2$ while the cost of factorization is on the order of $\kappa_{\rm fac}=n^3$ \cite{wk_ma_outage}.
\end{itemize}

Thus the overall computation cost for solving \eqref{ipm} using IPM is on the order of $\ln(1/\epsilon)\sqrt{\beta(\mathcal{K})}\\ \times(\kappa_{\rm for} + \kappa_{\rm fac})$. Using these concepts, we can now analyze the computational complexity of problems \eqref{PrConst2}, \eqref{PrConstS}, and \eqref{PrConstD}. Note that in all three formulations, the number of decision variables ($n$ in (43)) is on the order of $N_{\rm T}^2$ (ignoring the slack variables). Let us first examine problem (24), which has $L$ LMI (trace) constraints of size $1$, $L$ SOC constraints (of size $2$), $L$ LMI constraints of size $N_{\rm T}$, $K$ LMI (trace) constraints of size $1$, $K$ SOC constraints, $K$ LMI constraints of size $N_{\rm T}$, and in (24h), $1$ LMI constraints of size $N_{\rm T}$, $2L$ LMI constraints of size $1$, $2K$ LMI constraints of size $1$. Thus the complexity of the BTI-based algorithm is on the order shown in the first row of Table~\ref{tab_cplx}. Similarly, the complexity of the ${\mathcal S}$-procedure based approach and the LDI-based approach can be quantified as shown in the second and the third row of Table~\ref{tab_cplx}, respectively.
\begin{table}[h!]
\centering \caption{Complexity analysis of the proposed approaches} \label{tab_cplx}
\begin{tabular}[t]{|l|l|}
\hline
Method & Complexity Order ($n = \mathcal{O}(N_{\rm T}^2)$)\\
\hline
BTI & $\begin{array}{l}\ln(1/\epsilon)\sqrt{N_{\rm T}(K+L+1+)}n[(K+L)(N_{\rm T}^3 + nN_{\rm T}^2 \\ +3n+3)+ (K+L)
(N_{\rm T}^2+N_{\rm T}+1)^2 + N_{\rm T}^3 \\ +nN_{\rm T}^2+n^2]\end{array}$\\
\hline
${\mathcal S}$-procedure & $\begin{array}{l}\ln(1/\epsilon)\sqrt{N_{\rm T}(K+L+1+)}n[(K+L)(N_{\rm T}+1)^3 \\+n(K+L)(N_{\rm T}+1)^2
+n(N_{\rm T}^2+K+L)+N_{\rm T}^3\\ +K+L]\end{array}$\\
\hline
LDI & $\begin{array}{l}\ln(1/\epsilon)\sqrt{N_{\rm T}(K+L+1+)}n[(N_{\rm T}^3+3K+3L)\\ +n(N_{\rm T}^2+3K+3L)
+(K+L)((N_{\rm T}^2+1)^2\\ +(N_{\rm T}+1)^2)+n^2]\end{array}$\\
\hline
\end{tabular}
\end{table}

From Table~\ref{tab_cplx}, it is straightforward to show that LDI-based method has the lowest computational complexity since it involves only SOC constraints, while the BTI-based approach has the highest computational cost since it involves a more complicated set of constraints. However, in terms of tightness, the ${\mathcal S}$-procedure based approach performs the worst, as evidenced by our numerical results.

\section{Robust SRM}\label{sec_srm}
In the robust power minimization problem considered in Section~\ref{sec_rpm}, attempt has been made to keep the transmit power as low as possible yet maintaining the predefined secrecy rate $R$, as well as the harvested energy, within the secrecy outage probability. However, in many practical wireless communication systems (e.g., secondary users' transmission in cognitive radio systems, small cell users' transmission in heterogeneous networks (HetNets)), the maximum allowable transmission power is limited to a certain level so as to keep the interference to other users below a given threshold. In those scenarios, SRC power minimization problem \eqref{PrConst1} may turn out to be infeasible \cite{sec_mimome, zheng_chu_miso_out} and the designer may need to re-adjust the secrecy rate requirement in order to find a feasible solution. With insufficient channel knowledge, finding out the appropriate secrecy rate requirement can be a tedious job. Instead, a more attractive problem formulation can be to find the maximum secrecy rate $R$ that can be achieved subject to the same outage constraints and the additional transmit power constraint. Thus, the SRM problem with outage constraints for a given maximum transmission power can be represented as follows:
\begin{subequations}\label{srcP1}
\begin{eqnarray}
\max_{{\bf Q}_{\rm I}, R} \!\!\!& &\!\!\! R \label{srcP1_o}\\
{\rm s.t.} \!\!\!& &\!\!\! {\rm Pr}\left[\min_i\,\, \left\{C_{\rm I}\left({\bf Q}_{\rm I}\right) - \hat{C}_{{\rm e},i}\left({\bf Q}_{\rm I}\right)\right\}^+\ge R\right]\nonumber\\
\!\!\!& &\!\!\! \qquad\qquad\qquad\qquad \ge 1 - p,\forall i,\label{srcP1_c1}\\
\!\!\!& &\!\!\! {\rm Pr}\left[\min_k \hat E_k \geq \eta_k\right] \ge 1 - q, \forall k,\label{srcP1_c2}\\
\!\!\!& &\!\!\! {\rm tr} \left({\bf Q}_{\rm I}\right) \le P_{\rm T}, \quad {\bf Q}_{\rm I} \succeq {\bf 0}, \label{srcP1_c3}
\end{eqnarray}
\end{subequations}
where $P_{\rm T}$ is the maximum available transmission power budget at the transmitter. This problem is not convex in terms of the outage constraints. To make this problem more tractable, we propose a two-stage optimization procedure as shown below:
\begin{subequations}\label{srcP2}
\begin{empheq}[left=\max\limits_R \empheqlbrace]{align}
%\begin{align}
%\begin{cases}
\max_{{\bf Q}_{\rm I}} &~~ R \label{srcP2_o}\\%
{\rm s.t.} & ~ {\rm Pr}\left[\min_i \left\{C_{\rm I}\left({\bf Q}_{\rm I}\right) - \hat{C}_{{\rm e},i}\left({\bf Q}_{\rm I}\right)\right\}^+\ge R\right]\nonumber\\
&~ \qquad\qquad\qquad\qquad\ge 1 - p,\forall i,\label{srcP2_c1}\\
&~~ {\rm Pr}\left[\min_k \hat E_k \geq \eta_k\right] \ge 1 - q, \forall k,\label{srcP2_c2}\\
&~~ {\rm tr} \left({\bf Q}_{\rm I}\right) \le P_{\rm T}, \quad {\bf Q}_{\rm I} \succeq {\bf 0}. \label{srcP2_c3}
%\end{cases}
%\end{align}
\end{empheq}
\end{subequations}

In the first stage, we solve the inner maximization problem of \eqref{srcP2} for any given feasible $R$. In the second stage, we perform a one-dimensional line search over $R$ that leads to the optimal solution of the problem \eqref{srcP1}. Note that even with given $R$, the problem is not tractable due to the probabilistic constraints. Hence we apply the safe approximation approaches derived in the previous section for the probabilistic constraints.

%\eqref{PrConst1}

\subsection{SRM Based on BTI}\label{subsec_srm_bti}
According to \emph{Lemma~\ref{lemm_bti}}, the BTI based convex restrictions for constraints \eqref{srcP2_c1} and \eqref{srcP2_c2} are given by \eqref{impli_bti_sec} and \eqref{impli_bti_eh}, respectively, for given $R$. Thus, the safe approximation for the inner maximization problem in \eqref{srcP2} is given by
\begin{subequations}\label{srm_bti}
\begin{align}
\max_{{\bf Q}_{\rm I}, \{\psi_i\}, \{\omega_i\}, \{\nu_k\}, \{\varphi_k\}} ~~ & ~~ R\label{srm_bti_o}\\
{\rm s.t.} %& ~~ {\rm tr}\left({\bf R}_{{\rm H},i}^{\frac{1}{2}}\left({\bf I}_{N_{{\rm e},i}}\!\otimes\!{\bf Q}_{\rm I}\right){\bf R}_{{\rm H},i}^{\frac{1}{2}}\right) + \hat {\bf h}_{{\rm e},i}^H\left({\bf I}_{N_{{\rm e},i}}\otimes{\bf Q}_{\rm I}\right)\hat {\bf h}_{{\rm e},i} + \sqrt{-2\ln(\rho)}\psi_i \nonumber\\
%& \qquad\qquad - \ln(\rho)\omega_i - \frac{\sigma_{\rm e}^2}{2^{R}\sigma_{\rm d}^2}{\bf h}^H{\bf Q}_{\rm I}{\bf h} - \frac{\sigma_{\rm e}^2}{2^{R}} + {\sigma_{\rm e}^2} \le 0\label{srm_bti_c1}\\
& ~~ {\rm tr} \left({\bf Q}_{\rm I}\right) \le P_{\rm T},\label{srm_bti_c1}\\
& ~~ \mbox{$\eqref{PrConst2_c1}\text{--}\eqref{PrConst2_c7}$ satisfied.}\label{srm_bti_c2}%\\
%& \left\|\left[\begin{array}{c}{\rm vec}\left({\bf R}_{{\rm H},i}^{\frac{1}{2}}\left({\bf I}_{N_{{\rm e},i}}\!\otimes\!{\bf Q}_{\rm I}\right){\bf R}_{{\rm H},i}^{\frac{1}{2}}\right)\\ \sqrt{2}{\bf R}_{{\rm H},i}^{\frac{1}{2}}\left({\bf I}_{N_{{\rm e},i}}\!\otimes\!{\bf Q}_{\rm I}\right)\hat {\bf h}_{{\rm e},i}\end{array}\right]\right\|  \le \psi_i\label{PrConst2_c2}\\
%& \omega_i{\bf I}_{N_{\rm T}N_{{\rm e},i}} - {\bf R}_{{\rm H},i}^{\frac{1}{2}}\left({\bf I}_{N_{{\rm e},i}}\!\otimes\!{\bf Q}_{\rm I}\right){\bf R}_{{\rm H},i}^{\frac{1}{2}} \succeq {\bf 0}\label{PrConst2_c3}\\
%& {\rm tr}\left({\bf R}_{{\rm g},k}^{\frac{1}{2}}{\bf Q}_{\rm I}{\bf R}_{{\rm g},k}^{\frac{1}{2}}\right) + \hat {\bf g}_{k}^H{\bf Q}_{\rm I}\hat {\bf g}_{k} - \sqrt{-2\ln(\varrho)}\nu_k + \ln(\varrho)\varphi_k - \frac{\eta_k}{\xi} \ge 0\label{PrConst2_c4}\\% + {\sigma_{\rm h}^2}
%& \left\|\left[\begin{array}{c}{\rm vec}\left({\bf R}_{{\rm g},k}^{\frac{1}{2}}{\bf Q}_{\rm I}{\bf R}_{{\rm g},k}^{\frac{1}{2}}\right)\\ \sqrt{2}{\bf R}_{{\rm g},k}^{\frac{1}{2}}{\bf Q}_{\rm I}\hat {\bf g}_{k}\end{array}\right]\right\| \le \nu_k \label{PrConst2_c5}\\
%& \varphi_k{\bf I}_{N_{\rm T}} + {\bf R}_{{\rm g},k}^{\frac{1}{2}}{\bf Q}_{\rm I}{\bf R}_{{\rm g},k}^{\frac{1}{2}} \succeq {\bf 0} \label{PrConst2_c6}\\
%& {\bf Q}_{\rm I} \succeq {\bf 0}, \quad \psi_i, \omega_i \ge 0, \forall i, \quad \nu_k, \varphi_k \ge 0, \forall k, \label{PrConst2_c7}
%\eqref{impli_bti_sec}, ~~ \eqref{impli_bti_eh}, ~~ {\bf Q}_{\rm I} \succeq {\bf 0}. \label{PrConst2_c2}
\end{align}
\end{subequations}
Our next endeavour is to establish a link between the optimal solutions of \eqref{PrConst2} and \eqref{srm_bti}. If we can prove that the the optimal solution of problem \eqref{PrConst2} is also optimal for problem \eqref{srm_bti}, then we can readily obtain the optimal solution to the inner maximization problem in \eqref{srcP2}. %Towards this end, we provide the following proposition.

\begin{proposition}\label{prop_equi_bti}
Any optimal solution to the power minimization problem \eqref{PrConst2} is also optimal to the problem \eqref{srm_bti} for identical specifications.
\end{proposition}

\begin{proof}
The proof is identical to that of \cite[Theorem~2]{jrnl_secrecy} and is thus omitted for brevity.
\end{proof}

By {\em Theorem~\ref{thm_rank_bti}}, we know that the solution to the power minimization problem \eqref{PrConst2} is rank-one optimal, so is the solution of the problem \eqref{srm_bti} -- we can immediately infer from {\em Proposition~\ref{prop_equi_bti}}. The remaining task is to find the optimal $R$ from the second stage of problem \eqref{srcP2}. A simple one-dimensional linear search (e.g., bisection or golden-section search) over $R$ to find the maximal $R$ that solves the feasibility problem \eqref{srcP2} is sufficient. The lower boundary of the search  is obviously $0$ due to the assumption $R>0$. The upper limit can be defined by assuming that the system vintages the highest secrecy rate at zero eavesdropping capacity i.e., at ${C}_{{\rm e},i} = 0, \forall i$. Thus, we obtain the upper search limit from \eqref{srcP1_c1} as
\begin{multline}
R\le \log \left(1 + \frac{1}{\sigma_{\rm d}^2}{\bf h}^H{\bf Q}_{\rm I}{\bf h}\right) \leq \log \left(1 + \frac{1}{\sigma_{\rm d}^2}{\rm tr}({\bf Q}_{\rm I})\|{\bf h}\|^2\right)\\
\leq \log \left(1 + \frac{P_{\rm T}}{\sigma_{\rm d}^2}\|{\bf h}\|^2\right).
\end{multline}
Note that the last inequality is derived using the sum power constraint ${\rm tr}({\bf Q}_{\rm I})\leq P_{\rm T}$ in \eqref{srcP1_c3}.

\subsection{$\mathcal{S}$-Procedure Based SRM}\label{subsec_srm_sproc}
According to \emph{Lemma~\ref{lemm_sproc}}, the $\mathcal{S}$-procedure based convex restrictions for constraints \eqref{srcP2_c1} and \eqref{srcP2_c2} are given by \eqref{sec_sproc_lmi} and \eqref{eh_sproc_lmi}, respectively. As a result, a safe approximation for the inner maximization problem in \eqref{srcP2} is given by
\begin{subequations}\label{srm_sproc}
\begin{align}
\max_{{\bf Q}_{\rm I}, \{\mu_{{\rm H},i}\}, \{\mu_{{\rm g},k}\}} ~~ & ~~ R\label{srm_sproc_o}\\
{\rm s.t.} & ~~ {\rm tr} \left({\bf Q}_{\rm I}\right) \le P_{\rm T},\label{srm_sproc_c1}\\
& ~~ \mbox{$\eqref{PrConstS_c1}\text{--}\eqref{PrConstS_c3}$ satisfied.}\label{srm_sproc_c2}
\end{align}
\end{subequations}
The following proposition establishes the solution equivalence between \eqref{srm_sproc} and \eqref{PrConstS}.

\begin{proposition}
Any optimal solution to the power minimization problem \eqref{PrConstS} is also optimal to the problem \eqref{srm_sproc} for identical specifications.
\end{proposition}

\begin{proof}
The proof is identical to that of \cite[Theorem~2]{jrnl_secrecy} and is thus omitted for brevity.
\end{proof}

By {\em Theorem~\ref{thm_rank_sproc}}, it is already known that the optimal solution to the power minimization problem \eqref{PrConstS} yields a rank-one transmit covariance ${\bf Q}_{\rm I}$. Hence, the solution to the problem \eqref{srm_sproc} should also be of unit-rank. The remaining task is to find the optimal $R$ for problem \eqref{srcP2} through a one-dimensional line search as described in Subsection~\ref{subsec_srm_bti}.

\subsection{LDI Based SRM}\label{subsec_srm_ldi}
Similar to the BTI and the $\mathcal{S}$-procedure based approaches, a convex safe approximation for the inner maximization problem in \eqref{srcP2} is given by \emph{Lemma~\ref{lemm_ldi}} as
\begin{subequations}\label{srm_ldi}
\begin{align}
\max_{{\bf Q}_{\rm I}, \{\bar\psi_i\}, \{\bar\omega_i\}, \{\bar\nu_k\}, \{\bar\varphi_k\}} ~~ & ~~ R\label{srm_ldi_o}\\
{\rm s.t.} & ~~ {\rm tr} \left({\bf Q}_{\rm I}\right) \le P_{\rm T},\label{srm_ldi_c1}\\
& ~~ \mbox{$\eqref{PrConstD_c1}\text{--}\eqref{PrConstD_c7}$ satisfied,}\label{srm_ldi_c2}
\end{align}
\end{subequations}
which can be efficiently solved using existing solvers \cite{cvx}. The solution obtained is identical to that of the power minimization problem \eqref{PrConstD} as described by the following proposition.

\begin{proposition}
Any optimal solution to the power minimization problem \eqref{PrConstD} is also optimal to the problem \eqref{srm_ldi} for identical specifications.
\end{proposition}

\begin{proof}
The proof is identical to that of \cite[Theorem~2]{jrnl_secrecy} and is thus omitted for brevity.
\end{proof}

The solution has been proved to be rank-one optimal in {\em Theorem~\ref{thm_rank_ldi}}. Finally, a line search is performed to find the optimal $R$ as described in Subsection~\ref{subsec_srm_bti}.

\section{Simulation Results}\label{sec_sim}
Here, we study the performance of the proposed algorithms in MISO secrecy SWIPT systems with probabilistic constraints through numerical simulations. For simplicity, it was assumed that $\eta_k = \eta, \xi_k = 1,~\forall k$, $N_{{\rm e},i} = N_{\rm e},~\forall i$, $p = q = \rho$, $\sigma_{\rm d}^2 = \sigma_{\rm e}^2 = 1$, $d_{\rm I} = d_{{\rm h},k} = d_{{\rm e},i} = d$. We set $L_c = 35.97 \times 10^{-4}$, (i.e., when $G_T=18$ dBi, $G_R=-2$ dBi, $c=3\times 10^8, f=1000$ MHz), path loss exponent of $\bar\kappa = 2.7$, and distance $d = 10$ meters. In particular, we examine the case in which the QoS requirements are such that each user is provided with an outage probability of at most $10\%$; i.e., $\rho = 0.1$, unless otherwise specified. We simulated a flat Rayleigh fading environment where the channel vectors have entries with zero mean and variance $1/N_{\rm T}$.

Since Monte-Carlo simulations become prohibitively expensive under very low outage requirements for the convex restriction approaches, the algorithms developed in this paper are generally not recommended \cite{wk_ma_outage}. However, to illustrate the performance of the proposed approaches in a larger error domain, we average the results over $500$ realizations of the estimated channels.
%It should be noted that convex restriction methods do not require Monte-Carlo sampling, say, for rate outage verification or optimization purposes \cite{wk_ma_outage}.

We start the performance analysis of the proposed convex restriction formulations by comparing their feasibility rates, i.e., the chance of getting a feasible solution to the problem \eqref{PrConst1} in $500$ realizations of the estimated channels. Fig.~\ref{fig_feas} shows the feasibility rates of the three approaches for $N_{\rm T} = 6$, $K = L = N_{\rm e} = 3$, and $\eta = 0$ (dB). Interestingly, the $\mathcal{S}$-procedure based approach has the lowest feasibility rate compared to the other two methods. The BTI-based and the LDI-based approaches yield almost the same feasibility rate. Therefore, one would expect the worst performance from the $\mathcal{S}$-procedure based approach in terms of secrecy rate as well as transmit power which we will observe in the remaining examples.
\begin{figure}%[h]
\centering
\includegraphics*[width=8cm]{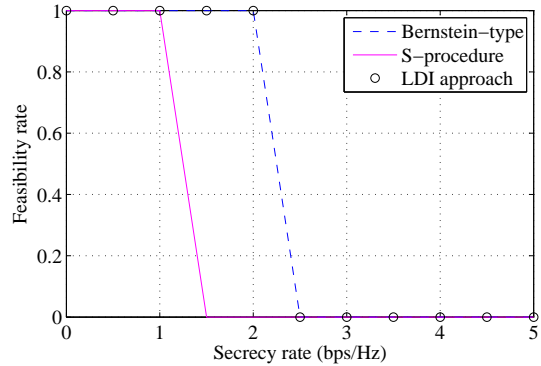}
\caption{Feasibility rates of the proposed methods with $N_{\rm T} = 6$, $K = L = N_{\rm e} = 3$, and $\eta = 0$ (dB).}\label{fig_feas}
\end{figure}

In the next example, we compare the transmit power consumptions of the proposed robust convex restriction solutions with that of the conventional perfect eavesdroppers' CSI based non-robust approach \cite{qli_sdp} for problem \eqref{PrConst1}. In the non-robust scheme, the estimated CSIs have been used for designing the transmit covariance matrix. Fig.~\ref{fig_Pt_R} shows the transmit power required by various methods against the rate outage threshold $R$ with $N_{\rm T} = 8$, $K = 3$, $L = 2$, $N_{\rm e} = 2$. We plot the results for $\eta = 10$ (dBm) and $\eta = 15$ (dBm). As can be seen from Fig.~\ref{fig_Pt_R}, the BTI-based approach yields slightly better transmit power performance compared to the LDI-based restriction approach. However, the $\mathcal{S}$-procedure based approach performs the worst amongst the proposed approaches. Interestingly, the robust algorithms guarantee the probabilistic constraints costing very little additional transmit power compared to the non-robust scheme. Obviously, with the increase in the required secrecy rate, all the algorithms demand higher transmit power.

\begin{figure}%[h]
\centering
\includegraphics*[width=8cm]{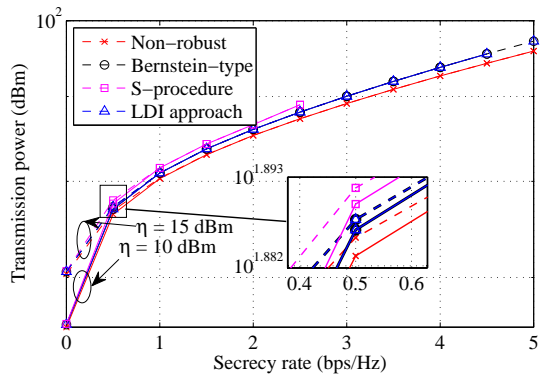}
\caption{Transmit power versus secrecy outage requirement $R$ with $N_{\rm T} = 8$, $K = 3$, $L = N_{\rm e} = 2$, and $\eta = 10, 15$ (dBm).}\label{fig_Pt_R}
\end{figure}

\begin{figure}%[h]
\centering
\includegraphics*[width=8cm]{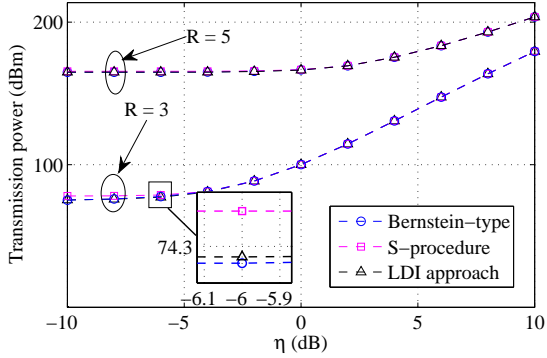}
\caption{Transmit power versus EH outage requirement $\eta$ with $N_{\rm T} = 8$, $K = 3$, $L = N_{\rm e} = 2$, and $R = 3, 5$ (bps/Hz).}\label{fig_Pt_eta}
\end{figure}

The results in Fig.~\ref{fig_Pt_R} indicate that when the EH requirement is increased from $\eta = -10$ (dB) to $\eta = 0$ (dB), all the approaches require more transmit power to satisfy the more demanding constraints. Hence, in the next example, we illustrate the transmit power consumptions of the solutions against the EH requirements. Fig.~\ref{fig_Pt_eta} shows the required transmit power versus the EH outage threshold $\eta$ with $N_{\rm T} = 8$, $K = 3$, $L = 2$, $N_{\rm e} = 2$. We plot the results for $R = 3$ and $R = 5$ (bps/Hz). The results in Fig.~\ref{fig_Pt_eta} show identical characteristics of the probabilistic  restriction solutions as those in Fig.~\ref{fig_Pt_R}. Again, with increased EH outage requirement, all the algorithms demand higher transmit power.

Next, we examine the secrecy rate performance of the proposed safe approximation approaches for the SRM problem \eqref{srcP1}. Fig.~\ref{fig_Cs_Pt} plots the worst-user secrecy rates of the various methods against the transmit power constraint $P_{\rm T}$ for outage tolerance of $5\%$ and $10\%$ for both the secrecy rate and EH constraints. As a baseline scheme, we also plot the secrecy rate of the classic maximal ratio transmission (MRT) scheme. In the MRT scheme, we pick the IR's channel direction for transmit beamforming assuming no eavesdropper is present. Specifically, the transmit covariance matrix is defined as ${\bf Q}_{\rm I} = \left(\frac{P_{\rm T}}{\|{\bf h}\|^2}\right){\bf h}{\bf h}^H$. Hence, it is a very suboptimal method and undergoes severe performance degradation compared with the proposed robust schemes. We set other parameters as $N_{\rm T} = 8$, $K = L = 3, N_{\rm e} = 2$, and $\eta = -5$ (dB). As we can see, the BTI- and LDI-based approaches yield almost identical worst-user secrecy rate whereas the $\mathcal{S}$-procedure based approach has noticeable degradation. Since the secrecy rate $C_{\rm s}$ in \eqref{cs} is an increasing function of the transmit power ${\rm tr}({\bf Q}_{\rm I})$, one can notice identical reflections in the results of Fig.~\ref{fig_Cs_Pt} with increasing $P_{\rm T}$. Also, it is no surprise that stricter outage constraint ($\rho = 0.05$) guarantees higher secrecy rate compared to relaxed outage constraint ($\rho = 0.1$). However, this observation does not apply to the conventional MRT scheme since it does not satisfy the outage constraints.

\begin{figure}%[h]
\centering
\includegraphics*[width=8cm]{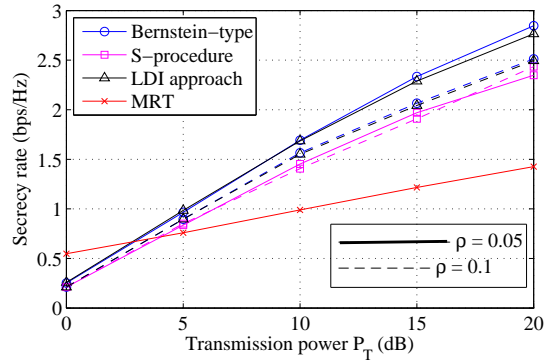}
\caption{Achievable secrecy rate by the proposed algorithms versus $P_{\rm T}$ with $N_{\rm T} = 8$, $K = L = N_{\rm e} = 3$, and $\eta = -5$.}\label{fig_Cs_Pt}
\end{figure}

\begin{figure}%[h]
\centering
\includegraphics*[width=8cm]{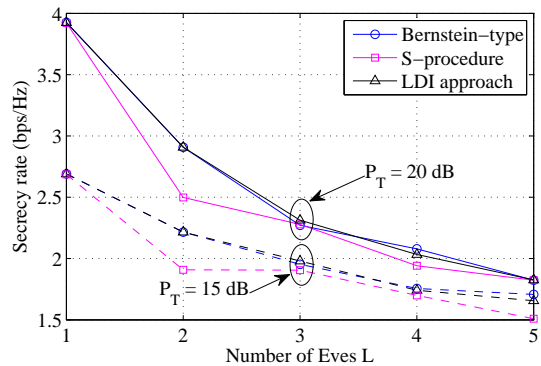}
\caption{The secrecy rate for different number of the eavesdroppers with $N_{\rm T} = 6$, $K = 3$, $N_{\rm e} = 2$, $P_{\rm T} = 15, 20$ (dB), and $\eta = -5$ (dB).}\label{fig_Cs_L}
\end{figure}

Finally, we analyze the achievable worst-user secrecy rate behaviors of the various methods when  the number of eavesdroppers is increased. The achievable secrecy rate versus the number of eavesdroppers (i.e., $L$) is shown in Fig.~\ref{fig_Cs_L} with $N_{\rm T} = 6$, $K = 3$, $N_{\rm e} = 2$, $P_{\rm T} = 15$ (dB), and $\eta = -5$ (dB). It can be observed from this result that the worst-user secrecy rate decreases as more eavesdroppers are present which is reasonable since it is more likely to find an Eve with stronger channel (resulting in worse secrecy) among an increased number of Eves. As we observed in the previous examples, the BTI- and LDI-based schemes outperform the $\mathcal{S}$-Procedure based one in terms of the achievable secrecy rate.

\section{Conclusions}\label{sec_con}
This paper investigated the power minimization as well as the SRM problems with probabilistic QoS constraints in MISOME systems for SWIPT and proposed convex safe approximation based transmit beamforming algorithms with imperfect CSI. Applying SDR techniques, we showed that rank-one optimal transmit covariance solutions are always obtainable for the safe approximation approaches. In particular, we found the maximum achievable secrecy rate under transmit power and outage constraints through optimally solving the power minimization problem. Simulation results have been provided to demonstrate the performance of the proposed approaches.

\appendix
\renewcommand{\thesubsection}{\Alph{subsection}}

\subsection{Proof of Theorem~\ref{thm_rank_bti}}\label{proof_thm_rank_bti}
We start the proof by eliminating ${\bf Q}_{\rm I}$ from the constraints \eqref{PrConst2_c2} and \eqref{PrConst2_c5} such that these two constraints become irrelevant to proof of rank of ${\bf Q}_{\rm I}$. Note that the SOC program (SOCP) constraint \eqref{PrConst2_c2} can be equivalently expressed as
\begin{align}
&\sqrt{\left\|{\bf R}_{{\rm H},i}^{\frac{1}{2}}\left({\bf I}_{N_{{\rm e},i}}\!\otimes\!{\bf Q}_{\rm I}\right){\bf R}_{{\rm H},i}^{\frac{1}{2}}\right\|_F^2 + 2\left\|{\bf R}_{{\rm H},i}^{\frac{1}{2}}\left({\bf I}_{N_{{\rm e},i}}\!\otimes\!{\bf Q}_{\rm I}\right)\hat {\bf h}_{{\rm e},i}\right\|^2}\nonumber\\
&\le \sqrt{\left\|{\bf R}_{{\rm H},i}^{\frac{1}{2}}\left({\bf I}_{N_{{\rm e},i}}\!\otimes\!{\bf Q}_{\rm I}\right)\right\|_F^2\left(\left\|{\bf R}_{{\rm H},i}^{\frac{1}{2}}\right\|_F^2 + 2\left\|\hat {\bf h}_{{\rm e},i}\right\|^2\right)}\nonumber\\
%&= \sqrt{{\rm tr}\left(\left({\bf I}_{N_{{\rm e},i}}\!\otimes\!{\bf Q}_{\rm I}\right)\left({\bf I}_{N_{{\rm e},i}}\!\otimes\!{\bf Q}_{\rm I}\right)^H\right)}\nonumber\\
%& \qquad\qquad\times \sqrt{\left({\rm tr}^2\left({\bf R}_{{\rm H},i}\right) + 2{\rm tr}\left({\bf R}_{{\rm H},i}\right)\left\|\hat {\bf h}_{{\rm e},i}\right\|^2\right)}\le \psi_i\nonumber\\
&\Longrightarrow {\rm tr}\left(\left({\bf I}_{N_{{\rm e},i}}\!\otimes\!{\bf Q}_{\rm I}\right)\left({\bf I}_{N_{{\rm e},i}}\!\otimes\!{\bf Q}_{\rm I}\right)^H\right) \le \frac{\psi_i^2}{\alpha_i},\label{sec_socp}
\end{align}
where $\alpha_i \triangleq {{\rm tr}^2\left({\bf R}_{{\rm H},i}\right) + 2{\rm tr}\left({\bf R}_{{\rm H},i}\right)\|\hat {\bf h}_{{\rm e},i}\|^2}$.

Applying the matrix trace identities ${\rm tr}\left[({\bf A}\otimes{\bf B})({\bf C}\otimes{\bf D})\right] = {\rm tr}({\bf AC}\otimes{\bf BD})$, ${\rm tr}({\bf A}\otimes{\bf B}) = {\rm tr}({\bf A}){\rm tr}({\bf B})$ \cite{mat_ana} and \eqref{mat_id3}, we obtain ${\rm tr}\left(\left({\bf I}_{N_{{\rm e},i}}\!\otimes\!{\bf Q}_{\rm I}\right)\left({\bf I}_{N_{{\rm e},i}}\!\otimes\!{\bf Q}_{\rm I}\right)^H\right) = {\rm tr}\left({\bf I}_{N_{{\rm e},i}}\otimes{\bf Q}_{\rm I}{\bf Q}_{\rm I}^H\right) = N_{{\rm e},i}{\rm tr}\left({\bf Q}_{\rm I}{\bf Q}_{\rm I}^H\right)$. Thus, \eqref{sec_socp} can be rewritten as
\begin{align}
{\rm tr}\left({\bf Q}_{\rm I}{\bf Q}_{\rm I}^H\right) \le \frac{\psi_i^2}{\alpha_i N_{{\rm e},i}}
%& \Rightarrow \lambda_{\max}\left({\bf Q}_{\rm I}{\bf Q}_{\rm I}^H\right) \le {\rm tr}\left({\bf Q}_{\rm I}{\bf Q}_{\rm I}^H\right) \le \frac{\psi_i^2}{\alpha N_{{\rm e},i}}\\
& \Rightarrow {\bf Q}_{\rm I}{\bf Q}_{\rm I}^H \le \kappa_i^2{\bf I}_{N_{\rm T}}\nonumber\\
& \Rightarrow \begin{bmatrix}\kappa_i{\bf I}_{N_{\rm T}} & {\bf Q}_{\rm I}\\
                            {\bf Q}_{\rm I}^H & \kappa_i{\bf I}_{N_{\rm T}}\end{bmatrix} \succeq {\bf 0}, \label{sec_socp2}
\end{align}
where $\kappa_i^2 \triangleq \frac{\psi_i^2}{\alpha_i N_{{\rm e},i}}$. Now we rewrite the constraint \eqref{sec_socp2} as the following LMI:
\begin{eqnarray}
\begin{bmatrix}\kappa_i{\bf I}_{N_{\rm T}} & {\bf 0}\\
               {\bf 0} & \kappa_i{\bf I}_{N_{\rm T}}\end{bmatrix}
               + \begin{bmatrix}{\bf 0} & {\bf Q}_{\rm I}\\
               {\bf Q}_{\rm I}^H & {\bf 0}\end{bmatrix} \succeq {\bf 0},
\end{eqnarray}
which can be eventually rewritten as
\begin{eqnarray}
\begin{bmatrix}\kappa_i{\bf I}_{N_{\rm T}} & {\bf 0}\\
               {\bf 0} & \kappa_i{\bf I}_{N_{\rm T}}\end{bmatrix}
               \!\!\!& \succeq &\!\!\! \begin{bmatrix}{\bf I}_{N_{\rm T}}\\{\bf 0}\end{bmatrix}  {\bf Q}_{\rm I} \begin{bmatrix}{\bf 0} & -{\bf I}_{N_{\rm T}}\end{bmatrix}
               + \begin{bmatrix}{\bf 0} \\ -{\bf I}_{N_{\rm T}}\end{bmatrix} {\bf Q}_{\rm I}^H\nonumber\\
               \!\!\!& &\!\!\! \times \begin{bmatrix}{\bf I}_{N_{\rm T}} & {\bf 0}\end{bmatrix}, \Vert{\bf Q}_{\rm I}\Vert \le \kappa_i. \label{sec_socp3}
\end{eqnarray}
The following lemma transforms the LMI in \eqref{sec_socp3} into a more convenient form.

\begin{lemma}[\bfseries Nemirovski lemma]\label{lemm_nemirovski}
\cite[Lemma 2]{nemirovski_lemma}
Given matrices ${\bf A}$, ${\bf B}$, ${\bf C}$ with ${\bf A} = {\bf A}^H$,
$${\bf A}\succeq{\bf B}^H{\bf X}{\bf C}+{\bf C}^H{\bf X}^H{\bf B},\quad \forall~{\bf X}:\Vert{\bf X}\Vert \le \kappa,$$
if and only if there exists a $\beta \ge 0$ such that
$$\begin{bmatrix}{\bf A}-\beta{\bf C}^H{\bf C}&-\kappa{\bf B}^H\\-\kappa{\bf B}& \beta{\bf I}\end{bmatrix}\succeq {\bf 0}.$$
\end{lemma}

%%%%%%%%%%%%Single column equation%%%%%%%%%%%%%%%%%%%%%%%%%%%%%%%%%%%%%%
\begin{figure*}[tb]
{%\footnotesize
\begin{eqnarray}
\begin{bmatrix}
\begin{bmatrix}\kappa_i{\bf I}_{N_{\rm T}} & {\bf 0}\\
               {\bf 0} & \kappa_i{\bf I}_{N_{\rm T}}\end{bmatrix}
               - \beta_1\begin{bmatrix}{\bf 0} \\ -{\bf I}_{N_{\rm T}}\end{bmatrix}\begin{bmatrix}{\bf 0} & -{\bf I}_{N_{\rm T}}\end{bmatrix} & \qquad-\kappa_i \begin{bmatrix}{\bf I}_{N_{\rm T}}\\{\bf 0}\end{bmatrix}\\
               -\kappa_i \begin{bmatrix}{\bf I}_{N_{\rm T}} & {\bf 0}\end{bmatrix} & \qquad\beta_1 {\bf I}_{N_{\rm T}}
\end{bmatrix} \succeq {\bf 0}.\label{sec_socp4}
\end{eqnarray}}

{%\footnotesize
\begin{eqnarray}
\begin{bmatrix}
\begin{bmatrix}\varphi_k{\bf I}_{N_{\rm T}} & {\bf 0}\\
               {\bf 0} & \varphi_k{\bf I}_{N_{\rm T}}\end{bmatrix}
               - \beta_2\begin{bmatrix}{\bf 0} \\ -{\bf I}_{N_{\rm T}}\end{bmatrix}\begin{bmatrix}{\bf 0} & -{\bf I}_{N_{\rm T}}\end{bmatrix} & \qquad-\varphi_k \begin{bmatrix}{\bf I}_{N_{\rm T}}\\{\bf 0}\end{bmatrix}\\
               -\varphi_k \begin{bmatrix}{\bf I}_{N_{\rm T}} & {\bf 0}\end{bmatrix} & \qquad\beta_2 {\bf I}_{N_{\rm T}}
\end{bmatrix} \succeq {\bf 0}.\label{eh_socp}
\end{eqnarray}}
\hrulefill \normalsize
\end{figure*}
%\setcounter{equation}{37}
%%%%%%%%%%%%%%%%%%%%%%%%%%%%%%%%%%%%%%%%%%%%%%%%%%%%%%%%%%%%%%%%%%%%%%%%
The merit of \emph{Lemma~\ref{lemm_nemirovski}} is that it transforms matrix inequalities into LMIs which do not involve the matrix variable. Based on \emph{Lemma~\ref{lemm_nemirovski}}, we obtain \eqref{sec_socp4} (at the top of the next page) from \eqref{sec_socp3}.
Similarly, the SOCP constraint \eqref{PrConst2_c5} can be rewritten as \eqref{eh_socp}, where $\varphi_k^2 \triangleq \frac{\nu_k^2}{{{\rm tr}^2\left({\bf R}_{{\rm g},k}\right) + 2{\rm tr}\left({\bf R}_{{\rm g},k}\right)\left\|\hat {\bf g}_{k}\right\|^2}}$. Indeed, the SOCP constraints \eqref{PrConst2_c2} and \eqref{PrConst2_c5} have been rewritten without ${\bf Q}_{\rm I}$ in \eqref{sec_socp4} and \eqref{eh_socp}, respectively.

The Lagrangian dual function of problem \eqref{PrConst2} is given by
\begin{multline}
\mathcal{L}\left({\bf Q}_{\rm I}, {\bf Z}, \lambda_i, {\bf C}_i, \mu_k, {\bf D}_k\right) \triangleq
%{\rm tr}\left({\bf Q}_{\rm I}\right) - {\rm tr}\left({\bf ZQ}_{\rm I}\right)\\
%+ \sum_{i=1}^L\lambda_i\left({\rm tr}\left({\bf R}_{{\rm H},i}^{\frac{1}{2}}\left({\bf I}_{N_{{\rm e},i}}\!\otimes\!{\bf Q}_{\rm I}\right){\bf R}_{{\rm H},i}^{\frac{1}{2}}\right) + \hat {\bf h}_{{\rm e},i}^H\left({\bf I}_{N_{{\rm e},i}}\otimes{\bf Q}_{\rm I}\right)\hat {\bf h}_{{\rm e},i}\right.\\
%\left. + \sqrt{-2\ln(p)}\psi_i - \ln(p)\omega_i - \frac{\sigma_{\rm e}^2}{2^R\sigma_{\rm d}^2}{\bf h}^H{\bf Q}_{\rm I}{\bf h} - \frac{\sigma_{\rm e}^2}{2^R} + {\sigma_{\rm e}^2}\right)\\
%-\sum_{i=1}^L{\rm tr}\left({\bf C}_i\left(\omega_i{\bf I}_{N_{{\rm e},i}} - {\bf R}_{{\rm H},i}^{\frac{1}{2}}\left({\bf I}_{N_{{\rm e},i}}\!\otimes\!{\bf Q}_{\rm I}\right){\bf R}_{{\rm H},i}^{\frac{1}{2}}\right)\right) - \sum_{k=1}^K\mu_k\\
%\times \left({\rm tr}\left({\bf R}_{{\rm g},k}^{\frac{1}{2}}{\bf Q}_{\rm I}{\bf R}_{{\rm g},k}^{\frac{1}{2}}\right) + \hat {\bf g}_{k}^H{\bf Q}_{\rm I}\hat {\bf g}_{k} + \sqrt{-2\ln(q)}\nu_k - \ln(q)\varphi_k\right.\\
%\left. - \frac{\eta_k}{\xi_k}\right) -\sum_{k=1}^K{\rm tr}\left({\bf D}_k\left(\varphi_k{\bf I}_{N_{\rm T}} + {\bf R}_{{\rm g},k}^{\frac{1}{2}}{\bf Q}_{\rm I}{\bf R}_{{\rm g},k}^{\frac{1}{2}}\right)\right)\\ =
{\rm tr}\left({\bf Q}_{\rm I}\right) - {\rm tr}\left({\bf ZQ}_{\rm I}\right) + \sum_{i=1}^L\sum_{n=1}^{N_{{\rm e},i}}\lambda_i\\
\times{\rm tr}\left(\boldsymbol\Psi_i^{(n,n)}{\bf Q}_{\rm I}\right)
+ \sum_{i=1}^L\sum_{n=1}^{N_{{\rm e},i}}{\rm tr}\left(\boldsymbol\Phi_i^{(n,n)}{\bf Q}_{\rm I}\right) + \sum_{i=1}^L\!\lambda_i \left(- \frac{\sigma_{\rm e}^2}{2^R\sigma_{\rm d}^2}\right.\\
\left. \times {\rm tr}({\bf h}{\bf h}^H{\bf Q}_{\rm I}) + \sqrt{-2\ln(p)}\psi_i - \ln(p)\omega_i - \frac{\sigma_{\rm e}^2}{2^R} + {\sigma_{\rm e}^2}\right)\\
-\sum_{i=1}^L{\rm tr}\left(\omega_i{\bf C}_i\right)
- \sum_{k=1}^K\mu_k{\rm tr}\left({\bf R}_{{\rm g},k}{\bf Q}_{\rm I} + \hat {\bf g}_{k}\hat {\bf g}_{k}^H{\bf Q}_{\rm I}\right)
-\sum_{k=1}^K{\rm tr}\left(\varphi_k\right.\\
\left.\times {\bf D}_k + {\bf R}_{{\rm g},k}^{\frac{1}{2}}{\bf D}_k{\bf R}_{{\rm g},k}^{\frac{1}{2}}{\bf Q}_{\rm I}\right) + \sum_{k=1}^K\mu_k\left(\sqrt{-2\ln(q)}\nu_k - \ln(q)\varphi_k\right.\\
\left. + \frac{\eta_k}{\xi_k}\right), \label{lagg_minP}
\end{multline}
where ${\bf Z}, \lambda_i, {\bf C}_i, \mu_k, {\bf D}_k$ are the Lagrangian dual variables associated with ${\bf Q}_{\rm I}$, \eqref{PrConst2_c1}, \eqref{PrConst2_c3}, \eqref{PrConst2_c4}, and \eqref{PrConst2_c6}, respectively. For notational convenience, $\boldsymbol\Psi_i^{(n,n)}$ and $\boldsymbol\Phi_i^{(n,n)}$ are also defined as the diagonal block sub-matrices of ${\bf R}_{{\rm H},i} + \hat{\bf h}_{{\rm e},i}\hat{\bf h}_{{\rm e},i}^H$ and ${\bf R}_{{\rm H},i}^{\frac{1}{2}}{\bf C}_i{\bf R}_{{\rm H},i}^{\frac{1}{2}}$, respectively. In particular,
${\bf R}_{{\rm H},i}\! + \hat{\bf h}_{{\rm e},i}\hat{\bf h}_{{\rm e},i}^H = \begin{bmatrix}\boldsymbol\Psi_i^{(1,1)} & \cdots & \boldsymbol\Psi_i^{(1,{N_{{\rm e},i}})}\\ \vdots & \ddots & \vdots\\ \boldsymbol\Psi_i^{({N_{{\rm e},i}},1)} & \cdots & \boldsymbol\Psi_i^{({N_{{\rm e},i}},{N_{{\rm e},i}})}\end{bmatrix}$ and ${\bf R}_{{\rm H},i}^{\frac{1}{2}}{\bf C}_i{\bf R}_{{\rm H},i}^{\frac{1}{2}} = \begin{bmatrix}\boldsymbol\Phi_i^{(1,1)} & \cdots & \boldsymbol\Phi_i^{(1,{N_{{\rm e},i}})}\\ \vdots & \ddots & \vdots\\ \boldsymbol\Phi_i^{({N_{{\rm e},i}},1)} & \cdots & \boldsymbol\Phi_i^{({N_{{\rm e},i}},{N_{{\rm e},i}})}\end{bmatrix}$.
The relevant Karush-Kuhn-Tucker (KKT) conditions can be defined as
\begin{subequations}\label{kkt_minP}
\begin{align}
& \frac{\partial\mathcal{L}\left({\bf Q}_{\rm I}, {\bf Z}, \lambda_i, {\bf C}_i, \mu_k, {\bf D}_k\right)}{\partial{\bf Q}_{\rm I}} = {\bf 0},\label{kkt1_minP}\\
& {\bf Z}{\bf Q}_{\rm I} = {\bf 0}, \quad {\bf Q}_{\rm I} \succeq {\bf 0},\label{kkt2_minP}\\
& {\bf Z} \succeq {\bf 0} ~~ \lambda_i \ge 0, ~~ {\bf C}_i \succeq {\bf 0}, \forall i, \mu_k \ge 0, {\bf D}_k \succeq {\bf 0}, \forall k.\label{kkt3_minP}
\end{align}
\end{subequations}
Next, we prove that there exists at least one $\lambda_i > 0$ such that $t \triangleq\frac{\sigma_{\rm e}^2\sum_{i=1}^L\lambda_i }{2^R\sigma_{\rm d}^2}$ is always positive definite. Denoting $\tau \triangleq \sum_{i=1}^L\lambda_i \left(- \frac{\sigma_{\rm e}^2}{2^R} + {\sigma_{\rm e}^2}\right) + \sum_{k=1}^K\frac{\mu_k\eta_k}{\xi_k}$ that includes the terms in the Lagrangian dual function \eqref{lagg_minP} not involving the primal variables, the dual problem of \eqref{PrConst2} is given by
\begin{subequations}\label{PrConst2Dual}
\begin{align}
\min_{{\bf Z}, \{\lambda_i\}, \{{\bf C}_i\}, \{\mu_k\}, \{{\bf D}_k\}} ~~ & \sum_{i=1}^L\lambda_i \left(\frac{\sigma_{\rm e}^2}{2^R} - {\sigma_{\rm e}^2}\right) - \sum_{k=1}^K\frac{\mu_k\eta_k}{\xi_k}\\
{\rm s. t.} ~~ & {\bf Z} \succeq {\bf 0} ~~ \lambda_i \ge 0, ~~ {\bf C}_i \succeq {\bf 0}, \forall i,\\
& \mu_k \ge 0, {\bf D}_k \succeq {\bf 0}, \forall k.
\end{align}
\end{subequations}
%The problem~\eqref{PrConst2Dual} can be equivalently rewritten as
%\begin{subequations}\label{PrConst2Dual2}
%\begin{align}
%\min_{{\bf Z}, \{\lambda_i\}, \{{\bf C}_i\}, \{\mu_k\}, \{{\bf D}_k\}} ~~ & ~~ \sum_{i=1}^L\lambda_i \left(\frac{\sigma_{\rm e}^2}{2^R} - {\sigma_{\rm e}^2}\right) - \sum_{k=1}^K\frac{\mu_k\eta_k}{\xi}\\% - {\sigma_{\rm h}^2}
%{\rm s. t.} ~~ &~~ {\bf Z} \succeq {\bf 0} ~~ \lambda_i \ge 0, ~~ {\bf C}_i \succeq {\bf 0}, \forall i,\\
%& ~~ \mu_k \ge 0, {\bf D}_k \succeq {\bf 0}, \forall k.
%\end{align}
%\end{subequations}
Note that the primal problem \eqref{PrConst2} is convex and it can be easily verified that the problem satisfies Slater's condition \cite{boyd}. Therefore, the duality gap is zero. Now, in order to successfully transfer information to the legitimate destination, the transmit power ${\rm tr} \left({\bf Q}_{\rm I}\right)$, which is the objective function of the primal problem \eqref{PrConst2}, must be greater than zero. Therefore, the strict positivity $\sum_{i=1}^L\lambda_i \left(\frac{\sigma_{\rm e}^2}{2^R} - {\sigma_{\rm e}^2}\right) - \sum_{k=1}^K\frac{\mu_k\eta_k}{\xi_k} > 0$ must also hold.

Let us now assume that $\lambda_i = 0, \forall i$. Then it can be easily observed that $\sum_{i=1}^L\lambda_i \left(\frac{\sigma_{\rm e}^2}{2^R} - {\sigma_{\rm e}^2}\right) - \sum_{k=1}^K\frac{\mu_k\eta_k}{\xi_k} \le 0$ for $\mu_k \ge 0, \forall k$, which contradicts with the already established fact. Thus we claim that there exists at least one $i$ for which $\lambda_i > 0$ and hence $t > 0$ must also hold.

Now, according to KKT condition \eqref{kkt1_minP}, we obtain
\begin{multline}\label{kkt12_minP}
%{\bf I}_{N_{\rm T}} - {\bf Z} + \sum_{i=1}^L\sum_{n=1}^{N_{{\rm e},i}}\lambda_i\boldsymbol\Psi_i^{(n,n)} + \sum_{i=1}^L\sum_{n=1}^{N_{{\rm e},i}}\boldsymbol\Phi_i^{(n,n)} - \sum_{i=1}^L\lambda_i \frac{\sigma_{\rm e}^2}{2^R\sigma_{\rm d}^2}{\bf h}{\bf h}^H - \sum_{k=1}^K\mu_k\left({\bf R}_{{\rm g},k} + \hat {\bf g}_{k}\hat {\bf g}_{k}^H\right)\\
% - \sum_{k=1}^K{\bf R}_{{\rm g},k}^{\frac{1}{2}}{\bf D}_k{\bf R}_{{\rm g},k}^{\frac{1}{2}} = {\bf 0}\\ \Rightarrow
{\bf I}_{N_{\rm T}} + \sum_{i=1}^L\sum_{n=1}^{N_{{\rm e},i}}\lambda_i\boldsymbol\Psi_i^{(n,n)} + \sum_{i=1}^L\sum_{n=1}^{N_{{\rm e},i}}\boldsymbol\Phi_i^{(n,n)} - \sum_{i=1}^L\lambda_i \frac{\sigma_{\rm e}^2}{2^R\sigma_{\rm d}^2}{\bf h}{\bf h}^H\\
- \sum_{k=1}^K\mu_k\left({\bf R}_{{\rm g},k} + \hat {\bf g}_{k}\hat {\bf g}_{k}^H\right)
 - \sum_{k=1}^K{\bf R}_{{\rm g},k}^{\frac{1}{2}}{\bf D}_k{\bf R}_{{\rm g},k}^{\frac{1}{2}} = {\bf Z}.
\end{multline}
Let ${\bf X} \triangleq {\bf I}_{N_{\rm T}} + \sum_{i=1}^L\sum_{n=1}^{N_{{\rm e},i}}\lambda_i\boldsymbol\Psi_i^{(n,n)} + \sum_{i=1}^L\sum_{n=1}^{N_{{\rm e},i}}\boldsymbol\Phi_i^{(n,n)} -\sum_{k=1}^K\mu_k\left({\bf R}_{{\rm g},k} + \hat {\bf g}_{k}\hat {\bf g}_{k}^H\right) - \sum_{k=1}^K{\bf R}_{{\rm g},k}^{\frac{1}{2}} {\bf D}_k{\bf R}_{{\rm g},k}^{\frac{1}{2}}$. Thus ${\bf Z} = {\bf X} - t{\bf h}{\bf h}^H$. From KKT condition \eqref{kkt2_minP}, we have ${\bf Z}{\bf Q}_{\rm I} = {\bf 0}$. Furthermore, it can be verified that in order to meet the secrecy rate constraints, it must hold that ${\bf Q}_{\rm I} \ne {\bf 0}$, or equivalently, ${\rm rank}\left({\bf Q}_{\rm I}\right) \ge 1$. Then from \eqref{kkt2_minP}, it follows that ${\rm rank}\left({\bf Z}\right) \le N_{\rm T} - 1$.

A key step in the proof of rank-one solution with secrecy constraints is proving that the matrix ${\bf X}$ is positive definite. Note that for secrecy problems without the EH constraints, the last two negative terms in the expression of ${\bf X}$ do not appear, and hence is readily established, see e.g., in \cite{zheng_chu_miso_out}. In contrast, for problems with EH constraints, this is a key challenging step which we overcome as follows.

Let $r_{\rm X} \triangleq {\rm rank}\left({\bf X}\right)$ denote the rank of ${\bf X}$. According to \cite[Lemma~5]{jrnl_secrecy}, it holds true that ${\rm rank}\left({\bf A}-{\bf B}\right) \geq {\rm rank}({\bf A}) - {\rm rank}({\bf B})$ for two matrices $\bf A$ and $\bf B$ of the same dimension. Thus ${\rm rank}\left({\bf Z}\right) \ge {\rm rank}\left({\bf X}\right) - {\rm rank}\left(t{\bf h}{\bf h}^H\right) = r_{\rm X} - 1$.

If ${\bf X}$ is positive-definite, $r_{\rm X} = N_{\rm T}$ and ${\rm rank}({\bf Z}) \geq N_{\rm T}-1$. However, if ${\rm rank}({\bf Z}) = N_{\rm T}$, i.e., ${\bf Z}$ is of full-rank, then it follows from \eqref{kkt2_minP} that ${\bf Q}_{\rm I} = {\bf 0}$, which cannot be an optimal solution to \eqref{PrConst2}. Therefore, we have ${\rm rank}({\bf Z}) = N_{\rm T}-1$. According to \eqref{kkt2_minP}, we have ${\rm rank}\big({\bf Q}_{\rm I}\big) = 1$.

For the case when $r_{\rm X} < N_{\rm T}$, let $\boldsymbol\Pi = \left[{\boldsymbol\pi}_1, {\boldsymbol\pi}_2, \dots, {\boldsymbol\pi}_{N_{\rm T}-r_{\rm X}}\right]$ with ${\boldsymbol\Pi}^H{\boldsymbol\Pi} = {\bf I}_{N_{\rm T}-r_{\rm X}}$ denote the orthogonal basis for the null space of ${\bf X}$, i.e., ${\bf X}{\boldsymbol\Pi}={\bf 0}$. Then we have
\begin{eqnarray}\label{eigA}
{\boldsymbol\pi}_i^H{\bf Z}{\boldsymbol\pi}_i \!\!\!&=&\!\!\! {\boldsymbol\pi}_i^H\left({\bf X} - t{\bf h}{\bf h}^{H}\right){\boldsymbol\pi}_i = -t|{\bf h}^{H}{\boldsymbol\pi}_i|^2\leq 0,\nonumber\\
\!\!\!& &\!\!\! \mbox{for }i = 1, \dots, N_{\rm T}-r_{\rm X}.
\end{eqnarray}
To guarantee that the Lagrangian in \eqref{lagg_minP} is bounded from below such that the dual function exists, it follows that ${\bf Z \succeq 0}$. Since ${\bf Z \succeq 0}$, it follows from \eqref{eigA} that ${\bf h}^{H}{\boldsymbol\pi}_i = 0, \forall i$, given $t>0$. That is,
\begin{equation}\label{no_info}
{\bf h}{\bf h}^{H}{\boldsymbol\Pi}={\bf 0}.
\end{equation}
As a result, we have
\begin{equation}\label{Apsi}
{\bf Z}{\boldsymbol\Pi}={\bf 0}.
\end{equation}
However, no information will be transferred to the IR in this case since all ${\boldsymbol\pi}_i$'s lie in the null space of ${\bf h}{\bf h}^{H}$ according to \eqref{no_info} \cite[Proof of Proposition~ 4.1]{rui_secrecy_swipt}. Hence positive secrecy rate constraint in \eqref{PrConst1_c1} cannot be satisfied in this case, which contradicts with the assumption $R>0$. Thus $r_{\rm X} = N_{\rm T}$ must hold, which implies that ${\bf X}$ is positive definite and ${\rm rank}\left({\bf Z}\right) \ge N_{\rm T} - 1$.

%We prove that ${\bf X} \succ {\bf 0}$ by contradiction. Let us assume that there exists some vector ${\bf z} \ne {\bf 0}$ such that ${\bf z}^H{\bf X}{\bf z} \le 0$. Thus we obtain ${\bf z}^H{\bf Z}{\bf z} = {\bf z}^H\left({\bf X}- t{\bf h}{\bf h}^H\right){\bf z} < 0$ for $t>0$. However, ${\bf z}^H{\bf Z}{\bf z} < 0$ contradicts with ${\bf Z} \succeq {\bf 0}$. Hence any such ${\bf z}$ does not exist and ${\bf X} \succ {\bf 0}$ must hold. Note that the extreme scenario of ${\bf h}{\bf h}^H = {\bf 0}$ has been carefully ignored since it can be easily verified from \eqref{PrConst1_c1} that positive secrecy rate is not achievable in such an extreme case.

Combining this result with ${\rm rank}\left({\bf Z}\right) \le N_{\rm T} - 1$, it follows that ${\rm rank}\left({\bf Z}\right) = N_{\rm T} - 1$. Accordingly, from \eqref{kkt2_minP}, we have
\begin{align}
{\rm rank}\left({\bf Q}_{\rm I}\right) = {\rm nullity}({\bf Z}) = N_{\rm T} - {\rm rank}\left({\bf Z}\right) = 1.
\end{align}
Theorem~\ref{thm_rank_bti} is thus proved. \hfill$\blacksquare$

\subsection{Proof of Theorem~\ref{thm_rank_sproc}} \label{proof_thm_rank_sproc}
By defining
\begin{align}
{\boldsymbol\Sigma}_{{\rm H},i} & \!\triangleq\! \left[\!\!\begin{array}{cc}{\bf R}_{{\rm H},i}^{\frac{1}{2}}\left({\bf I}_{N_{{\rm e},i}}\!\otimes\!{\bf Q}_{\rm I}\right){\bf R}_{{\rm H},i}^{\frac{1}{2}} & {\bf R}_{{\rm H},i}^{\frac{1}{2}}\left({\bf I}_{N_{{\rm e},i}}\!\otimes\!{\bf Q}_{\rm I}\right)\hat {\bf h}_{{\rm e},i}\\
\hat {\bf h}_{{\rm e},i}^H\left({\bf I}_{N_{{\rm e},i}}\!\otimes\!{\bf Q}_{\rm I}\right){\bf R}_{{\rm H},i}^{\frac{1}{2}} & \hat {\bf h}_{{\rm e},i}^H\left({\bf I}_{N_{{\rm e},i}}\otimes{\bf Q}_{\rm I}\right)\hat {\bf h}_{{\rm e},i} \end{array}\!\right],\nonumber\\
{\boldsymbol\Lambda}_{{\rm H},i} &\triangleq \left[\begin{array}{cc}\mu_{{\rm H},i}{\bf I}_{N_{\rm T}N_{{\rm e},i}} & {\bf 0}\\
{\bf 0} & \frac{\sigma_{\rm e}^2}{2^R} - {\sigma_{\rm e}^2} - \mu_{{\rm H},i}\gamma_{{\rm e}}^2\end{array}\right]\nonumber,
\end{align}
${\boldsymbol\Theta}_i\left({\bf Q}_{\rm I},\mu_{{\rm H},i}\right)$ in \eqref{sec_sproc_lmi} can be decomposed as
\begin{align}
{\boldsymbol\Theta}_i\left({\bf Q}_{\rm I},\mu_{{\rm H},i}\right) & =
%-{\boldsymbol\Sigma}_{{\rm H},i} + {\boldsymbol\Lambda}_{{\rm H},i} + \frac{\sigma_{\rm e}^2}{2^R\sigma_{\rm d}^2}\left[{\bf 0} ~~ {\bf h}\right]^H{\bf Q}_{\rm I}\left[{\bf 0} ~~ {\bf h}\right]\nonumber\\
%& =
-\left[{\bf R}_{{\rm H},i}^{\frac{1}{2}}~~\hat {\bf h}_{{\rm e},i}\right]^H\left({\bf I}_{N_{{\rm e},i}}\!\otimes\!{\bf Q}_{\rm I}\right)\left[{\bf R}_{{\rm H},i}^{\frac{1}{2}}~~\hat {\bf h}_{{\rm e},i}\right]\nonumber\\
& \qquad + {\boldsymbol\Lambda}_{{\rm H},i} + \frac{\sigma_{\rm e}^2}{2^R\sigma_{\rm d}^2}\left[{\bf 0} ~~ {\bf h}\right]^H{\bf Q}_{\rm I}\left[{\bf 0} ~~ {\bf h}\right].
\end{align}
Similarly, denoting ${\boldsymbol\Sigma}_{{\rm g},k} \triangleq \left[\begin{array}{cc}{\bf R}_{{\rm g},k}^{\frac{1}{2}}{\bf Q}_{\rm I}{\bf R}_{{\rm g},k}^{\frac{1}{2}} & {\bf R}_{{\rm g},k}^{\frac{1}{2}}{\bf Q}_{\rm I}\hat {\bf g}_{k}\\
\hat {\bf g}_{k}^H{\bf Q}_{\rm I}{\bf R}_{{\rm g},k}^{\frac{1}{2}} & \hat {\bf g}_{k}^H{\bf Q}_{\rm I}\hat {\bf g}_{k}\end{array}\right]$ and ${\boldsymbol\Lambda}_{{\rm g},k} \triangleq \left[\begin{array}{cc}\mu_{{\rm g},k}{\bf I}_{N_{\rm T}} & {\bf 0}\\
{\bf 0} & \frac{\eta_k}{\xi_k} - \mu_{{\rm g},k}\gamma_{{\rm e}}^2\end{array}\right],$ we obtain ${\boldsymbol\Upsilon}_k\left({\bf Q}_{\rm I},\mu_{{\rm g},k}\right) = {\boldsymbol\Sigma}_{{\rm g},k} + {\boldsymbol\Lambda}_{{\rm g},k} = \left[{\bf R}_{{\rm g},k}^{\frac{1}{2}}~~\hat {\bf g}_k\right]^H{\bf Q}_{\rm I}\left[{\bf R}_{{\rm g},k}^{\frac{1}{2}}~~\hat {\bf g}_k\right] + {\boldsymbol\Lambda}_{{\rm g},k}$ from \eqref{eh_sproc_lmi}. Let us now define the Lagrangian dual function of problem \eqref{PrConstS}:
\begin{multline}
\mathcal{L}\left({\bf Q}_{\rm I}, {\bf Z}, {\bf V}_{{\rm H},i}, {\bf V}_{{\rm g},k}\right) \triangleq
%{\rm tr}\left({\bf Q}_{\rm I}\right)
%- {\rm tr}\left({\bf ZQ}_{\rm I}\right) - \sum_{i=1}^L{\rm tr}\left({\bf V}_{{\rm H},i}\right.\\
%\left. \quad\times{\boldsymbol\Theta}_i\left({\bf Q}_{\rm I},\mu_{{\rm H},i}\right)\right)
%- \sum_{k=1}^K{\rm tr}\left({\bf V}_{{\rm g},k}{\boldsymbol\Upsilon}_k\left({\bf Q}_{\rm I},\mu_{{\rm g},k}\right)\right)\\% + {\sigma_{\rm h}^2}
%= {\rm tr}\left({\bf Q}_{\rm I}\right) - {\rm tr}\left({\bf ZQ}_{\rm I}\right) + \sum_{i=1}^L{\rm tr}\left({\bf V}_{{\rm H},i}\left[{\bf R}_{{\rm H},i}^{\frac{1}{2}}~~\hat {\bf h}_{{\rm e},i}\right]^H\left({\bf I}_{N_{{\rm e},i}}\right.\right.\\
%\left.\left.\otimes{\bf Q}_{\rm I}\right)\left[{\bf R}_{{\rm H},i}^{\frac{1}{2}}~~\hat {\bf h}_{{\rm e},i}\right]\right) - \sum_{i=1}^L{\rm tr}\left({\bf V}_{{\rm H},i}{\boldsymbol\Lambda}_{{\rm H},i}\right)
%- \sum_{k=1}^K{\rm tr}\left({\bf V}_{{\rm g},k}{\boldsymbol\Sigma}_{{\rm g},k}\right)\\
%- \sum_{k=1}^K{\rm tr}\left({\bf V}_{{\rm g},k}{\boldsymbol\Lambda}_{{\rm g},k}\right) - \frac{\sigma_{\rm e}^2}{2^R\sigma_{\rm d}^2}\sum_{i=1}^L{\rm tr}\left({\bf V}_{{\rm H},i}\left[{\bf 0} ~~ {\bf h}\right]^H{\bf Q}_{\rm I}\left[{\bf 0} ~~ {\bf h}\right]\right)\\
{\rm tr}\left({\bf Q}_{\rm I}\right) - {\rm tr}\left({\bf ZQ}_{\rm I}\right)\\
+ \sum_{i=1}^L\sum_{n=1}^{N_{{\rm e},i}}{\rm tr}\left(\boldsymbol\Omega_i^{(n,n)}{\bf Q}_{\rm I}\right) - \sum_{i=1}^L{\rm tr}\left({\bf V}_{{\rm H},i}{\boldsymbol\Lambda}_{{\rm H},i}\right) - \sum_{k=1}^K{\rm tr}\left({\bf V}_{{\rm g},k}\right.\\
\left.\times{\boldsymbol\Lambda}_{{\rm g},k}\right)
- \sum_{k=1}^K{\rm tr}\left({\bf V}_{{\rm g},k}\left[{\bf R}_{{\rm g},k}^{\frac{1}{2}}~~\hat {\bf g}_k\right]^H{\bf Q}_{\rm I}\left[{\bf R}_{{\rm g},k}^{\frac{1}{2}}~~\hat {\bf g}_k\right]\right)\\
- \frac{\sigma_{\rm e}^2}{2^R\sigma_{\rm d}^2}\sum_{i=1}^L{\rm tr}\left({\bf V}_{{\rm H},i}\left[{\bf 0} ~~ {\bf h}\right]^H{\bf Q}_{\rm I}\left[{\bf 0} ~~ {\bf h}\right]\right),\label{lagg_minPS}
%- \sum_{k=1}^K\mu_k{\rm tr}\left({\bf R}_{{\rm g},k}{\bf Q}_{\rm I} + \hat {\bf g}_{k}\hat {\bf g}_{k}^H{\bf Q}_{\rm I}\right)\\
%-\sum_{k=1}^K{\rm tr}\left(\varphi_k{\bf D}_k + {\bf R}_{{\rm g},k}^{\frac{1}{2}}{\bf D}_k{\bf R}_{{\rm g},k}^{\frac{1}{2}}{\bf Q}_{\rm I}\right) + \sum_{k=1}^K\mu_k\left(\sqrt{-2\ln(\varrho)}\nu_k - \ln(\varrho)\varphi_k - \frac{\eta_k}{\xi}\right) % + {\sigma_{\rm h}^2}
\end{multline}
where ${\bf Z}, {\bf V}_{{\rm H},i}, {\bf V}_{{\rm g},k}$ are the Lagrangian dual variables associated with ${\bf Q}_{\rm I}$, \eqref{PrConstS_c1}, and \eqref{PrConstS_c2}, respectively. For notational convenience, $\boldsymbol\Omega_i^{(n,n)} \in \mathbb{H}_+^{N_{\rm T}}$ is defined as the diagonal block sub-matrices of $\left[{\bf R}_{{\rm H},i}^{\frac{1}{2}}~~\hat {\bf h}_{{\rm e},i}\right]{\bf V}_{{\rm H},i}\left[{\bf R}_{{\rm H},i}^{\frac{1}{2}}~~\hat {\bf h}_{{\rm e},i}\right]^H$. In particular,
\begin{multline}
\left[{\bf R}_{{\rm H},i}^{\frac{1}{2}}~~\hat {\bf h}_{{\rm e},i}\right]{\bf V}_{{\rm H},i}\left[{\bf R}_{{\rm H},i}^{\frac{1}{2}}~~\hat {\bf h}_{{\rm e},i}\right]^H\\
= \begin{bmatrix}\boldsymbol\Omega_i^{(1,1)} & \cdots & \boldsymbol\Omega_i^{(1,{N_{{\rm e},i}})}\\ \vdots & \ddots & \vdots\\ \boldsymbol\Omega_i^{({N_{{\rm e},i}},1)} & \cdots & \boldsymbol\Omega_i^{({N_{{\rm e},i}},{N_{{\rm e},i}})}\end{bmatrix}.
\end{multline}
The KKT conditions relevant to the proof can be defined as
\begin{subequations}\label{kkt_minPS}
\begin{align}
\frac{\partial\mathcal{L}\left({\bf Q}_{\rm I}, {\bf Z}, {\bf V}_{{\rm H},i}, {\bf V}_{{\rm g},k}\right)}{\partial{\bf Q}_{\rm I}} & = {\bf 0},\label{kkt1_minPS}\\
{\bf Z}{\bf Q}_{\rm I} & = {\bf 0}, ~~ {\bf Q}_{\rm I} \succeq {\bf 0},\label{kkt2_minPS}\\
{\bf Z} \succeq {\bf 0} ~~ {\bf V}_{{\rm H},i} & \succeq {\bf 0}, \forall i, {\bf V}_{{\rm g},k} \succeq {\bf 0}, \forall k.\label{kkt3_minPS}
\end{align}
\end{subequations}
Next, we prove that ${\bf T} \triangleq\frac{\sigma_{\rm e}^2}
%there is at least one ${\bf V}_{{\rm H},i} \succ {\bf 0}$ such that
{2^R\sigma_{\rm d}^2}\sum_{i=1}^L {\bf V}_{{\rm H},i}$ is always positive definite by contradiction. Denoting $\tau_{\rm S} \triangleq \sum_{i=1}^L{\rm tr}\left({\bf V}_{{\rm H},i}\left[\begin{array}{cc}{\bf 0} & {\bf 0}\\
{\bf 0} & \frac{\sigma_{\rm e}^2}{2^R} - {\sigma_{\rm e}^2}\end{array}\right]\right) - \sum_{k=1}^K{\rm tr}
\left({\bf V}_{{\rm g},k}\left[\begin{array}{cc}{\bf 0} & {\bf 0}\\
{\bf 0} & -\frac{\eta_k}{\xi_k}\end{array}\right]\right)$ that includes the terms in the
Lagrangian dual function not involving the primal variables, the dual problem of \eqref{PrConstS} is given by %\eqref{lagg_minPS}
\begin{subequations}\label{PrConst2DualS}
\begin{align}
\min_{{\bf Z}, \{{\bf V}_{{\rm H},i}\}, \{{\bf V}_{{\rm g},k}\}} ~ & ~ -\tau_{\rm S}\\
{\rm s. t.} ~ &~ {\bf Z} \succeq {\bf 0}, ~ {\bf V}_{{\rm H},i} \succeq {\bf 0}, {\bf V}_{{\rm g},k} \succeq {\bf 0}, {\forall i,k}.
\end{align}
\end{subequations}
%The problem~\eqref{PrConst2Dual} can be equivalently rewritten as
%\begin{subequations}\label{PrConst2Dual2}
%\begin{align}
%\min_{{\bf Z}, \{\lambda_i\}, \{{\bf C}_i\}, \{\mu_k\}, \{{\bf D}_k\}} ~~ & ~~ \sum_{i=1}^L\lambda_i \left(\frac{\sigma_{\rm e}^2}{2^R} - {\sigma_{\rm e}^2}\right) - \sum_{k=1}^K\frac{\mu_k\eta_k}{\xi}\\% - {\sigma_{\rm h}^2}
%{\rm s. t.} ~~ &~~ {\bf Z} \succeq {\bf 0} ~~ \lambda_i \ge 0, ~~ {\bf C}_i \succeq {\bf 0}, \forall i,\\
%& ~~ \mu_k \ge 0, {\bf D}_k \succeq {\bf 0}, \forall k.
%\end{align}
%\end{subequations}
Note that the primal problem \eqref{PrConstS} is convex and it can be easily verified that the problem satisfies Slater's condition \cite{boyd}. Thus the duality gap is zero. Now, in order to successfully transfer information to the legitimate destination, the transmit power ${\rm tr} \left({\bf Q}_{\rm I}\right)$, which is the objective function of the primal problem \eqref{PrConstS}, must be greater than zero. Therefore, the strict positivity on the objective of the dual problem \eqref{PrConst2DualS} must also hold.

Now, if ${\bf V}_{{\rm H},i} = {\bf 0}, \forall i,$ resulting in ${\bf T} = {\bf 0}$, then it can be easily observed that $\tau_{\rm S} \le 0$ for ${\bf V}_{{\rm g},k} \succeq {\bf 0}, \forall k$, which contradicts with the already established fact. Thus we claim that ${\bf T} \succ {\bf 0}$ must also hold.

According to KKT condition \eqref{kkt1_minPS}, we obtain
\begin{equation}
%{\bf I}_{N_{\rm T}} - {\bf Z} + \!\!\sum_{i=1}^L\sum_{n=1}^{N_{{\rm e},i}}\boldsymbol\Omega_i^{(n,n)} - \sum_{k=1}^K\left[{\bf R}_{{\rm g},k}^{\frac{1}{2}}~~\hat {\bf g}_k\right]\!{\bf V}_{{\rm g},k}\!\left[{\bf R}_{{\rm g},k}^{\frac{1}{2}}~~\hat {\bf g}_k\right]^H\\
%- \left[{\bf 0} ~~ {\bf h}\right]{\bf T}\left[{\bf 0} ~~ {\bf h}\right]^H = {\bf 0}\label{kkt12_minPS}\\
%\Longrightarrow
{\bf X}_{\rm S} - \left[{\bf 0} ~~ {\bf h}\right]{\bf T}\left[{\bf 0} ~~ {\bf h}\right]^H = {\bf Z},
\end{equation}%{multline}
where ${\bf X}_{\rm S} \triangleq {\bf I}_{N_{\rm T}} + \sum_{i=1}^L\sum_{n=1}^{N_{{\rm e},i}}\boldsymbol\Omega_i^{(n,n)} - \sum_{k=1}^K\left[{\bf R}_{{\rm g},k}^{\frac{1}{2}}~~\hat {\bf g}_k\right]{\bf V}_{{\rm g},k}\left[{\bf R}_{{\rm g},k}^{\frac{1}{2}}~~\hat {\bf g}_k\right]^H$. From KKT condition \eqref{kkt2_minPS}, we have ${\bf Z}{\bf Q}_{\rm I} = {\bf 0}$. Furthermore, it can be verified that in order to meet the secrecy rate constraints, it must hold that ${\bf Q}_{\rm I} \ne {\bf 0}$, or equivalently, ${\rm rank}\left({\bf Q}_{\rm I}\right) \ge 1$. Then from \eqref{kkt2_minP}, it follows that ${\rm rank}\left({\bf Z}\right) \le N_{\rm T} - 1$.

Following similar reasoning as in Appendix~\ref{proof_thm_rank_bti}, it can be proved by contradiction that ${\bf X}_{\rm S} \succ {\bf 0}$.
%and thus ${\rm rank}\left({\bf X}_{\rm S}\right) = N_{\rm T}$ must hold.
%Now the ${\rm rank}\left(\left[{\bf 0} ~~ {\bf h}\right]{\bf T}\left[{\bf 0} ~~ {\bf h}\right]^H\right)$ is one since left/right multiplying any matrix by a positive definite matrix does not change its rank \cite{mat_ana}.
%According to \cite[Lemma~5]{jrnl_secrecy}, it holds true that ${\rm rank}\left({\bf A}-{\bf B}\right) \geq {\rm rank}({\bf A}) - {\rm rank}({\bf B})$ for two matrices $\bf A$ and $\bf B$ of the same dimension. Thus
Applying \cite[Lemma~5]{jrnl_secrecy}, ${\rm rank}\left({\bf Z}\right) \ge {\rm rank}\left({\bf X}_{\rm S}\right) - {\rm rank}\left(\left[{\bf 0} ~~ {\bf h}\right]{\bf T}\left[{\bf 0} ~~ {\bf h}\right]^H\right) = N_{\rm T} - 1$. Combining this argument with ${\rm rank}\left({\bf Z}\right) \le N_{\rm T} - 1$, it follows that ${\rm rank}\left({\bf Z}\right) = N_{\rm T} - 1$. Accordingly, from \eqref{kkt2_minP}, ${\rm rank}\left({\bf Q}_{\rm I}\right) = 1$. \hfill$\blacksquare$
% Theorem~\ref{thm_rank_sproc} is thus proved.

\subsection{Proof of Theorem~\ref{thm_rank_ldi}}\label{proof_thm_rank_ldi}
Applying \emph{Lemma~\ref{lemm_nemirovski}} (Nemirovski lemma), constraints \eqref{PrConstD_c2}, \eqref{PrConstD_c3}, \eqref{PrConstD_c5}, and \eqref{PrConstD_c6} can be equivalently transformed into LMIs not involving ${\bf Q}_{\rm I}$. Therefore, we ignore these constraints in the proof of rank-one ${\bf Q}_{\rm I}$ as in Appendix~\ref{proof_thm_rank_bti}. Accordingly, the Lagrangian of \eqref{PrConstD} is given by
\begin{multline}
\mathcal{L}\left({\bf Q}_{\rm I}, {\bf Z}, \boldsymbol\lambda\right) \triangleq
%{\rm tr}\left({\bf Q}_{\rm I}\right) - {\rm tr}\left({\bf ZQ}_{\rm I}\right) + \sum_{i=1}^L\lambda_{b,i}\left({\rm tr}\left({\bf R}_{{\rm H},i}^{\frac{1}{2}}\right.\right.\\
%\left. \times \left({\bf I}_{N_{{\rm e},i}}\!\otimes\!{\bf Q}_{\rm I}\right){\bf R}_{{\rm H},i}^{\frac{1}{2}}\right) + \hat {\bf h}_{{\rm e},i}^H\left({\bf I}_{N_{{\rm e},i}}\otimes{\bf Q}_{\rm I}\right)\hat {\bf h}_{{\rm e},i} - \frac{\sigma_{\rm e}^2}{2^R\sigma_{\rm d}^2}{\bf h}^H\\
%\left.\times{\bf Q}_{\rm I}{\bf h} - \frac{\sigma_{\rm e}^2}{2^R} + {\sigma_{\rm e}^2} + 2\sqrt{-\ln(p)}\left(\bar\psi_i + \bar\omega_i\right) \right)
%- \sum_{k=1}^K\lambda_{e,k}\\
%\left({\rm tr}\left({\bf R}_{{\rm g},k}^{\frac{1}{2}}{\bf Q}_{\rm I}{\bf R}_{{\rm g},k}^{\frac{1}{2}}\right) + \hat {\bf g}_{k}^H{\bf Q}_{\rm I}\hat {\bf g}_{k} - \frac{\eta_k}{\xi_k} - 2\sqrt{-\ln(q)}\right.\\
%\left.\qquad\qquad\qquad\qquad\qquad \times\left(\bar\nu_k + \bar\varphi_k\right)\right)\\ =
{\rm tr}\left({\bf Q}_{\rm I}\right) - {\rm tr}\left({\bf ZQ}_{\rm I}\right) + \sum_{i=1}^L\sum_{n=1}^{N_{{\rm e},i}}\lambda_{b,i}{\rm tr}\left(\boldsymbol\Psi_i^{(n,n)}{\bf Q}_{\rm I}\right)\\
+ \sum_{i=1}^L\lambda_{b,i} \left(- \frac{\sigma_{\rm e}^2}{2^R\sigma_{\rm d}^2}{\bf h}^H{\bf Q}_{\rm I}{\bf h} - \frac{\sigma_{\rm e}^2}{2^R} + {\sigma_{\rm e}^2} + 2\sqrt{-\ln(p)}\right.\\
\left.\times \left(\bar\psi_i + \bar\omega_i\right)\right)
- \sum_{k=1}^K\lambda_{e,k}{\rm tr}\left({\bf R}_{{\rm g},k}{\bf Q}_{\rm I} + \hat {\bf g}_{k}\hat {\bf g}_{k}^H{\bf Q}_{\rm I}\right)\\
- \sum_{k=1}^K\lambda_{e,k}\left(- \frac{\eta_k}{\xi_k} - 2\sqrt{-\ln(q)}\left(\bar\nu_k + \bar\varphi_k\right)\right), \label{lagg_minD}
\end{multline}
where $\boldsymbol\lambda \triangleq \left[\lambda_{b,1},\dots, \lambda_{b,L}, \lambda_{e,1}, \dots, \lambda_{e,K}\right]^T$ is the vector of Lagrange multipliers associated with the constraints \eqref{PrConstD_c1} and \eqref{PrConstD_c4}, and ${\bf Z}$ is the Lagrangian dual variable associated with ${\bf Q}_{\rm I}$. The KKT conditions can be defined as
%relevant to the proof
\begin{subequations}\label{kkt_minPD}
\begin{align}
\frac{\partial\mathcal{L}\left({\bf Q}_{\rm I}, {\bf Z}, \boldsymbol\lambda\right)}{\partial{\bf Q}_{\rm I}} & = {\bf 0},\label{kkt1_minPD}\\
{\bf Z}{\bf Q}_{\rm I} & = {\bf 0},~~
{\bf Q}_{\rm I} & \succeq {\bf 0}, ~~ {\bf Z} \succeq {\bf 0} ~~ \boldsymbol\lambda \ge 0.\label{kkt2_minPD}
\end{align}
\end{subequations}
Denoting $\tau_D \triangleq \sum_{i=1}^L\lambda_{b,i}\left(- \frac{\sigma_{\rm e}^2}{2^R} + {\sigma_{\rm e}^2}\right) + \sum_{k=1}^K\lambda_{e,k} \frac{\eta_k}{\xi_k}$ in \eqref{PrConst2Dual}, the dual problem can be defined as
\begin{subequations}\label{PrConst2DualD}
\begin{align}
\max_{{\bf Z}, \boldsymbol\lambda} ~~ & ~~ -\tau_D\\
{\rm s. t.} ~~ &~~ {\bf Q}_{\rm I},~{\bf Z} \succeq {\bf 0} ~~ \boldsymbol\lambda \ge 0.
\end{align}
\end{subequations}
Following similar arguments as in Appendix~\ref{proof_thm_rank_bti}, it can be shown that there exists at least one $\lambda_{b,i}>0$. The rest of the proof is identical to the corresponding part in Appendix~\ref{proof_thm_rank_bti}, and thus is omitted for brevity. \hfill$\blacksquare$

\bibliographystyle{IEEEtran}\footnotesize{
\bibliography{IEEEabrv,refdb}}%

% Generated by IEEEtran.bst, version: 1.13 (2008/09/30)
\begin{thebibliography}{10}
\providecommand{\url}[1]{#1}
\csname url@samestyle\endcsname
\providecommand{\newblock}{\relax}
\providecommand{\bibinfo}[2]{#2}
\providecommand{\BIBentrySTDinterwordspacing}{\spaceskip=0pt\relax}
\providecommand{\BIBentryALTinterwordstretchfactor}{4}
\providecommand{\BIBentryALTinterwordspacing}{\spaceskip=\fontdimen2\font plus
\BIBentryALTinterwordstretchfactor\fontdimen3\font minus
  \fontdimen4\font\relax}
\providecommand{\BIBforeignlanguage}[2]{{%
\expandafter\ifx\csname l@#1\endcsname\relax
\typeout{** WARNING: IEEEtran.bst: No hyphenation pattern has been}%
\typeout{** loaded for the language `#1'. Using the pattern for}%
\typeout{** the default language instead.}%
\else
\language=\csname l@#1\endcsname
\fi
#2}}
\providecommand{\BIBdecl}{\relax}
\BIBdecl

\bibitem{krikidis_swipt_mag}
I.~Krikidis, S.~Timotheou, S.~Nikolaou, G.~Zheng, D.~W.~K. Ng, and R.~Schober,
  ``Simultaneous wireless information and power transfer in modern
  communication systems,'' \emph{IEEE Commun. Magazine}, vol.~52, pp. 104--110,
  Nov. 2014.

\bibitem{swipt_bc}
R.~Zhang and C.~K. Ho, ``{MIMO} broadcasting for simultaneous wireless
  information and power transfer,'' \emph{IEEE Trans. Wireless Commun.},
  vol.~12, pp. 1989--200, May 2013.

\bibitem{swipt_archi}
X.~Zhou, R.~Zhang, and C.~Ho, ``Wireless information and power transfer:
  {A}rchitecture design and rate-energy tradeoff,'' \emph{IEEE Trans. Commun.},
  vol.~61, pp. 4757--4767, Nov. 2013.

\bibitem{swipt_oppor}
L.~Liu, R.~Zhang, and K.~C. Chua, ``Wireless information transfer with
  opportunistic energy harvesting,'' \emph{IEEE Trans. Wireless Commun.},
  vol.~12, pp. 288--300, Jan. 2013.

\bibitem{swipt_robust}
Z.~Xiang and M.~Tao, ``Robust beamforming for wireless information and power
  transmission,'' \emph{IEEE Wireless Commun. Letters}, vol.~1, pp. 372--375,
  Aug. 2012.

\bibitem{jrnl_swipt}
M.~R.~A. Khandaker and K.-K. Wong, ``{SWIPT} in {MISO} multicasting systems,''
  \emph{IEEE Wireless Commun. Letters}, vol.~3, pp. 277--280, June 2014.

\bibitem{conf_swipt}
------, ``{QoS}-based multicast beamforming for {SWIPT},'' in \emph{Proc. IEEE
  SECON Workshop Energy Harvesting Commun.}, Singapore, June 30-July 03, 2014,
  pp. 62-67.

\bibitem{swipt_mu}
J.~Xu, L.~Liu, and R.~Zhang, ``Multiuser {MISO} beamforming for simultaneous
  wireless information and power transfer,'' in \emph{Proc. IEEE Int. Conf.
  Acoust., Speech, Signal Process. (ICASSP)}, 2013, pp. 4754-4758.

\bibitem{rui_secrecy_swipt}
L.~Liu, R.~Zhang, and K.-C. Chua, ``Secrecy wireless information and power
  transfer with {MISO} beamforming,'' \emph{IEEE Trans. Signal Process.},
  vol.~62, pp. 1850--1863, Apr. 2014.

\bibitem{jrnl_secrecy}
M.~R.~A. Khandaker and K.-K. Wong, ``Masked beamforming in the presence of
  energy-harvesting eavesdroppers,'' \emph{IEEE Trans. Inf. Forensics and
  Security}, vol.~10, pp. 40--54, Jan. 2015.

\bibitem{jrnl_secrecy_sinr}
------, ``Robust secrecy beamforming with energy-harvesting eavesdroppers,''
  \emph{IEEE Wireless Commun. Letters}, vol.~4, pp. 10--13, Feb. 2015.

\bibitem{sec_swipt_ofdma}
M.~Zhang and Y.~Liu, ``Energy harvesting for physical-layer security in {OFDMA}
  networks,'' \emph{IEEE Trans. Info. Forensics Security}, vol.~11, pp.
  154--162, Jan. 2016.

\bibitem{sec_swipt_ofdma2}
M.~Zhang, Y.~Liu, and R.~Zhang, ``Artificial noise aided secrecy information
  and power transfer in {OFDMA} systems,'' \emph{IEEE Trans. Wireless Commun.},
  vol.~15, pp. 3085--3096, Apr. 2016.

\bibitem{rui_fading_secrecy}
H.~Xing, L.~Liu, and R.~Zhang, ``Secrecy wireless information and power
  transfer in fading wiretap channel,'' \emph{IEEE Trans. Vehicular
  Technology}, vol.~65, pp. 180--190, Jan. 2016.

\bibitem{wc_swipt_lb}
R.~Feng, Q.~Li, Q.~Zhang, and J.~Qin, ``Robust secure transmission in {MISO}
  simultaneous wireless information and power transfer system,'' \emph{IEEE
  Trans. Veh. Technol.}, vol.~64, pp. 400--405, Jan. 2015.

\bibitem{wc_swipt_mimo}
S.~Wang and B.~Wang, ``Robust secure transmit design in {MIMO} channels with
  simultaneous wireless information and power transfer,'' \emph{IEEE Signal
  Process. Letters}, vol.~22, pp. 2147--2151, Nov. 2015.

\bibitem{zheng_chu_swipt}
Z.~Chu, Z.~Zhu, M.~Johnston, and S.~L. Goff, ``Simultaneous wireless
  information power transfer for {MISO} secrecy channel,'' \emph{IEEE Trans.
  Veh. Technol.}, to appear, 2016.

\bibitem{arman_tit_sinr}
A.~Shojaeifard, K.~A. Hamdi, E.~Alsusa, D.~K.~C. So, and J.~Tang, ``Exact
  {SINR} statistics in the presence of heterogeneous interferers,'' \emph{IEEE
  Trans. Info. Theory}, vol.~61, pp. 6759--6773, Dec. 2015.

\bibitem{bti}
I.~Bechar, ``A {Bernstein}-type inequality for stochastic processes of
  quadratic forms of {Gaussian} variables,'' available online at
  \texttt{http://arxiv.org/abs/0909.3595}, Sep. 2009.

\bibitem{ben_tal_sproc}
A.~Ben-Tal and A.~Nemirovski, ``Robust solutions of linear programming problems
  contaminated with uncertain data,'' \emph{Math. Program.}, vol. 88, ser. A,
  pp. 411--424, 2000.

\bibitem{wk_ma_outage}
K.-Y. Wang, A.~M.-C. So, T.-H. Chang, W.-K. Ma, and C.-Y. Chi, ``Outage
  constrained robust transmit optimization for multiuser {MISO} downlinks:
  Tractable approximations by conic optimization,'' \emph{IEEE Trans. Signal
  Process.}, vol.~62, pp. 5690--5705, Nov. 2014.

\bibitem{zheng_chu_miso_out}
Z.~Chu, H.~Xing, M.~Johnston, and S.~L. Goff, ``Secrecy rate optimizations for
  a {MISO} secrecy channel with multiple multiantenna eavesdroppers,''
  \emph{IEEE Trans. Wireless Commun.}, vol.~15, pp. 283--297, Jan. 2016.

\bibitem{secrecy_swipt_rob}
D.~W.~K. Ng, E.~S. Lo, and R.~Schober, ``Robust beamforming for secure
  communication in systems with wireless information and power transfer,''
  \emph{IEEE Trans. Wireless Commun.}, vol.~13, pp. 4599--4615, Aug. 2014.

\bibitem{goel_an}
S.~Goel and R.~Negi, ``Guaranteeing secrecy using artificial noise,''
  \emph{IEEE Trans. Wireless Commun.}, vol.~7, pp. 2180--2189, June 2008.

\bibitem{qli_sdp}
Q.~Li and W.-K. Ma, ``Optimal and robust transmit designs for miso channel
  secrecy by semidefinite programming,'' \emph{IEEE Trans. Signal Process.},
  vol.~59, pp. 3799--3812, Aug. 2011.

\bibitem{cvx}
M.~Grant and S.~Boyd, ``{CVX}: Matlab software for disciplined convex
  programming (web page and software).'' \texttt{http://cvxr.com/cvx}, Apr.,
  2010.

\bibitem{boyd}
S.~Boyd and L.~Vandenberghe, \emph{Convex Optimization}.\hskip 1em plus 0.5em
  minus 0.4em\relax Cambridge, U.~K.: Cambridge University Press, 2004.

\bibitem{ldi}
S.~Janson, ``Large deviations for sums of partly dependent random variables,''
  \emph{Random Struct. Algorithms}, vol.~24, pp. 234--248, May 2004.

\bibitem{lmi_sdp}
S.-S. Cheung, A.~M.-C. So, and K.~Wang, ``Linear matrix inequalities with
  stochastically dependent perturbations and applications to chanceconstrained
  semidefinite optimization,'' \emph{SIAM J. Optim.}, vol.~22, pp. 1394--1430,
  2012.

\bibitem{cvx_opt_lect}
A.~Ben-Tal and A.~Nemirovski, \emph{Lectures on modern convex optimization:
  {A}nalysis, algorithms, and engineering applications}.\hskip 1em plus 0.5em
  minus 0.4em\relax Philadelphia, PA, USA: MPS SIAM Series on Optimization,
  2001.

\bibitem{cplx_sdp}
M.~X. Goemans and D.~P. Williamson, ``Approximation algorithms for {MAX-3-CUT}
  and other problems via complex semidefinite programming,'' \emph{J. Comput.
  Syst. Sci.}, vol.~68, pp. 442--470, 2004.

\bibitem{sec_mimome}
K.~Cumanan, Z.~Ding, B.~Sharif, G.~Tian, and K.~Leung, ``Secrecy rate
  optimizations for a {MIMO} secrecy channel with a multiple-antenna
  eavesdropper,'' \emph{IEEE Trans. Veh. Technol.}, vol.~63, pp. 1678--1690,
  May 2014.

\bibitem{mat_ana}
R.~A. Horn and C.~R. Johnson, \emph{Matrix Analysis}.\hskip 1em plus 0.5em
  minus 0.4em\relax Cambridge, {U. K.}: Cambridge Univ. Press, 1985.

\bibitem{nemirovski_lemma}
Y.~Eldar, A.~Ben-Tal, and A.~Nemirovski, ``Robust mean-squared error estimation
  in the presence of model uncertainties,'' \emph{IEEE Trans. Signal Process.},
  vol.~53, pp. 168--181, Jan. 2005.

\end{thebibliography}

\end{document}